\documentclass{amsproc}
\usepackage{amsmath,amssymb,cite,graphicx,epstopdf,epsfig,psfrag}
\newtheorem{theorem}{Theorem}[section]
\newtheorem{lemma}[theorem]{Lemma}
\newtheorem{corollary}[theorem]{Corollary}
\theoremstyle{definition}
\newtheorem{definition}[theorem]{Definition}
\newtheorem{example}[theorem]{Example}

\theoremstyle{remark}
\newtheorem{remark}[theorem]{Remark}
\numberwithin{equation}{section}

\newcommand{\Lu}{L\~u\hspace{-.25em}\raisebox{-.1 em}{'}}
\newcommand{\String}{\mathcal{L}}
\newcommand{\eps}{\varepsilon}
\newcommand{\N}{\mathbb N}
\newcommand{\Z}{\mathbb Z}
\newcommand{\Q}{\mathbb Q}
\newcommand{\R}{\mathbb R}
\newcommand{\Com}{\mathbb C}

\newcommand{\mcc}{\mathcal{C}}
\newcommand{\Dimensions}{\mathcal D}
\newcommand{\mcl}{\mathcal L}
\newcommand{\Rerror}{\mathcal R}
\newcommand{\mfr}{\mathfrak{R}}

\DeclareMathOperator{\res}{res}

\begin{document}

 \title[Minkowski dimension  and tube formulas for $p$-adic fractal strings]{Minkowski dimension  and explicit tube formulas for $p$-adic fractal strings} 

\author{Michel L. Lapidus}
\address{Department of Mathematics, University of California, Riverside, CA 92521-0135}
\email{lapidus@math.ucr.edu}

\thanks{The work of the first author (MLL) was partially supported by the US National Science Foundation (NSF) under the research grants DMS-0707524 and DMS-1107750, and by the Institut des Hautes Etudes Scientifiques (IHES) in Paris/Bures-sur-Yvette, France, where the first author was a visiting professor while part of this work was completed, 
as well as by the Burton Jones Endowed Chair in Pure Mathematics 
(of which MLL was the chair holder at the University of California, Riverside, during the completion of this paper). 
}

\author {L\~u' H\`ung  }
\address{Department of Mathematics, Hawai`i Pacific University, Honolulu, HI 96813-2785}
\email{l2h5ngt6m2@gmail.com}
\thanks{The research of the second author (LH) was partially supported by the Trustees' Scholarly Endeavor Program at Hawai`i Pacific University.}

\author {Machiel van Frankenhuijsen}
\address{Department of Mathematics, Utah Valley University, Orem, UT 84058-5999}
\email{vanframa@uvu.edu}


\date{\today}


\begin{abstract}
The theory of complex dimensions describes the oscillations in the geometry (spectra and dynamics) of fractal strings. Such geometric oscillations can be seen most clearly in the  explicit volume formula for the tubular neighborhoods of a $p$-adic fractal string 
$\mathcal{L}_p$, expressed in terms of the underlying complex dimensions.
The general fractal tube formula obtained in this paper is illustrated by several examples, including the nonarchimedean Cantor and Euler strings.  
Moreover, we show that the Minkowski dimension of a $p$-adic fractal string coincides with the abscissa of convergence of the geometric zeta function associated with the string, as well as with the asymptotic growth rate of the corresponding geometric counting function.
The proof of this new result can be applied to both real and $p$-adic fractal strings and hence, yields a unifying explanation of a key result in the theory of complex dimensions for fractal strings, even in the archimedean (or real) case. 
\end{abstract}

\maketitle

\tableofcontents

\begin{quote}
{\em
Nature is an infinite sphere of which the center is everywhere and the circumference nowhere.}
\hspace{\stretch{1}} Blaise Pascal (1623--1662)
\end{quote}

\section{Introduction} 

An ordinary real (or archimedean) fractal string is a bounded open subset of the real line, with fractal boundary.
It provides a complementary perspective to the notion of a self-similar fractal, in the sense that 
every self-similar string determines a self-similar set in $\mathbb R$, viewed as the boundary of the string. 
Moreover, it is noteworthy that the geometric zeta function of a fractal string determines the fractal (i.e., Minkowski) dimension of the corresponding fractal set. 
Following the examples of the $a$-string and of the archimedean Cantor string given by the first author in~\cite{Lap1, Lap, Lap3},
the notion of an archimedean fractal string was conceived and defined by Michel Lapidus and Carl Pomerance in their investigation of the one-dimensional Weyl--Berry conjecture for fractal drums and its connection  with the Riemann zeta function ~\cite{LapPo}.
The Riemann hypothesis turned out to be equivalent to the solvability of the corresponding inverse spectral problem for fractal strings, as was established by Michel Lapidus and Helmut Maier in ~\cite{LapMa}.
The heuristic notion of complex dimensions then started to emerge and was used in a crucial way, at least heuristically, in their spectral reformulation of the Riemann hypothesis.
The precise notion of complex dimensions, defined as the poles of a certain geometric zeta function associated with the fractal string, was crystallized and rigorously developed by Michel Lapidus and Machiel van Frankenhuijsen in the research monograph {\em Fractal Geometry and Number Theory\/} \cite{L-vF1} and then significantly further extended in the book
{\em Fractal Geometry, Complex Dimensions and Zeta Functions\/}~\cite{L-vF2}.

The work of Lapidus and Maier mentioned above can be summarized as follows:
\begin{quote}\em
The inverse spectral problem for a fractal string can be solved if and only if its dimension is not $1/2$.
\end{quote}
The inverse spectral problem for a fractal string is not solvable in dimension 1/2 because the Riemann zeta function $\zeta(s)=1+1/2^s+1/3^s+\cdots$ 
vanishes (infinitely often) on the critical line $\Re(s)=1/2$. Therefore, the inverse spectral problem  is solvable in dimension $D\neq 1/2$ if and only if the Riemann zeta function does not vanish off the critical line $\Re(s)=1/2$ or, equivalently, if and only if the Riemann hypothesis is true. 
 In order to understand why dimension $1/2$ is so special, Michel
Lapidus ~\cite{Lap2} (building, in particular, on earlier work in ~\cite{LapPo} and in ~\cite{LapMa}), as well as Michel Lapidus and Machiel van Frankenhuijsen \cite{L-vF2}, were led to a definition of the dual of a fractal string,
interchanging the dimensions $D$ and $1-D$. 
A partial answer to the question why $1/2$ is singled out is that if a fractal string is self-dual,
then its dimension is $1/2$.
The concept of the dual of a fractal string provides a geometric way to take the functional equation of the Riemann zeta function into consideration in the theory of complex dimensions.
Such considerations would not be complete if they did not also involve the Euler product of the Riemann zeta function.
The spectral operator was 
introduced semi-heuristically in~\cite{L-vF2} as the map that sends the geometry of a fractal string onto its spectrum. Formally, the spectral operator admits an (operator-valued) Euler product. In \cite{HerLap, HerLap2, HerLap3}, Hafedh Herichi and Michel Lapidus have rigorously defined and studied the spectral operator, within a proper functional analytic setting. They have also reformulated the above criterion for the Riemann hypothesis in terms of a suitable notion of invertibility of the spectral operator; see also \cite{Lap4, Lap5} for a corresponding asymmetric criterion, expressed in terms of the usual notion of invertibility of the spectral operator. Furthermore, they have shown that the (operator-valued or ``quantized'') Euler product for the spectral operator also converges {\em inside\/} the critical strip $0 < \Re(s) < 1,$ where all the (nontrivial) complex zeros of the Riemann zeta function reside.

In order to extend the framework of the theory of complex dimensions,
with an aim towards applying ideas and  techniques from number theory to the inverse spectral problem,
it is therefore natural to attempt developing a theory of $p$-adic fractal strings,
and then globally,
an ad\`elic theory of fractal strings. Further laying out the foundations for such a theory is one of the main long-term goals of this paper.

We note that  nonarchimedean $p$-adic  analysis has been used in various areas of mathematics, such as arithmetic geometry, number theory and representation theory, as well as  of mathematical and theoretical physics, such as string theory,  cosmology, quantum mechanics, relativity theory, quantum field theory, statistical and condensed matter physics; 
see, e.g., \cite{Drag, Dykkv, Ulam, RTV, VVZ,} and the relevant references therein.
(See also, e.g., \cite{Ben, BGPW, Ber} and \cite{Duc, PB}.)
In \emph{Number theory as the ultimate physical theory,}\cite{Vol},
Igor V. Volovich has suggested that $p$-adic numbers  can possibly be used in order to describe the geometry of spacetime at very high energies
 (and hence, very small scales, i.e., below the Planck or the string scale) 
  because the measurements  in the `archimedean' geometry of spacetime at fine scales do not have any certainty.
Furthermore,
several authors, including Stephen W. Hawking, have also suggested that the fine scale structure of spacetime may be fractal; see, e.g.,~\cite{GibHaw,HawIs,Lap2,Not,WheFo}.
Therefore, a geometric theory of $p$-adic fractal strings and their complex dimensions might be helpful in the  quest to explore the geometry and fine scale structure of spacetime at high energies. 

On the other hand, in the book entitled \emph{In Search of the Riemann Zeros: Strings, fractal membranes and noncommutative spacetimes} ~\cite{Lap2},
Michel Lapidus has suggested that fractal strings and their quantization, fractal membranes, may be related to aspects of string theory and that $p$-adic (and possibly, ad\`elic) analogs of these notions would be useful in this context in order to better understand the underlying noncommutative spacetimes and their moduli spaces ([\citen{Lap2,LapNe}]).  The theory of $p$-adic fractal strings, once suitably `quantized',  may be helpful in further developing some of these ideas    and eventually providing a framework for unifying the real and $p$-adic fractal strings and membranes.

In this paper, we further develop the geometric theory of $p$-adic (or nonarchimedean)  fractal strings, which are  bounded open subsets of the $p$-adic  $\Q_p$ with  a fractal ``boundary'', along with the associated theory of complex dimensions and, especially, of fractal tube formulas in the nonarchimedean setting.  
 This theory, which was first developed by Michel Lapidus and 
\Lu\ H\`ung  in~\cite{LapLu1,LapLu2,LapLu2.5}, as well as later, by those same authors and Machiel van Frankenhuijsen in~\cite{LapLu3},  extends the theory of real (or archimedean)  fractal strings and their complex dimensions  in a natural way.
Following~\cite{LapLu1,LapLu2, LapLu2.5,LapLu3},
we introduce suitable geometric zeta functions for $p$-adic fractal strings whose poles play the role of complex dimensions for the standard real fractal strings. 
We also discuss the explicit fractal tube formulas  in the general case of (languid) $p$-adic fractal strings and in the special case of $p$-adic self-similar strings. Throughout this paper, these various results are illustrated in the case of  the nonarchimedean self-similar Cantor and the Fibonacci strings (introduced in~\cite{LapLu1,LapLu2}), as well as in the case of the 
  $p$-adic Euler string (introduced in~\cite{LapLu2.5,LapLu3}), which (strictly speaking) is not self-similar. 
Some particular attention is devoted to the  3-adic Cantor string (introduced and studied in~\cite{LapLu1}),
 whose `metric' boundary is the  3-adic Cantor set~\cite{LapLu1}, which is naturally homeomorphic to the classic ternary Cantor set. 
 
The rest of this paper is organized as follows:
In \S\ref{pfs} and \S\ref{GZF},
 we recall the definition of an arbitrary $p$-adic fractal string
 along with its associated geometric zeta function and complex dimensions,
 as well as the corresponding notions of Minkowski dimension and content. Furthermore, in \S\ref{inner tube}, we discuss the important,  but more technical, question of how to suitably define and calculate the volume
of the `inner' $\eps$-neighborhood  of a $p$-adic fractal string. The definitions and proofs given in \S\ref{inner tube} depend crucially on the nonarchimedean (specifically, $p$-adic) nature of the underlying geometry.
In \S\ref{Mdimension}, we then introduce a proper notion in this context of Minkowski dimension and Minkowski content, building on the results of \S\ref{inner tube}. Moreover, we obtain the nonarchimedean analog of a key result in (archimedean) fractal string theory, showing that the Minkowski dimension of a $p$-adic fractal string coincides with the abscissa of convergence of its geometric zeta function. In the process, we show that this common value coincides with the asymptotic growth rate of the geometric counting function of the fractal string, a result which is new even in the archimedean setting, even though it may be implicit in \cite{LapPo}. Our proof also provides a new and unified derivation of the archimedean result for real fractal strings (\hspace{-1.3mm} \cite{Lap, Lap3}, \cite{L-vF1, L-vF2}), by placing the archimedean and the nonarchimedean settings on the same footing; see \S\ref{TRC}.
In \S\ref{explicit tf},
we then use our previous results (in \S\ref{inner tube} and \S\ref{Mdimension}) to express the volume of the inner tube of a $p$-adic fractal string as an infinite sum over all the underlying complex dimensions,
thereby obtaining a nonarchimedean analog of the `fractal tube formula' obtained for real (or archimedean) fractal strings in ~\cite{L-vF1,L-vF2}. (See, especially, \cite[Ch.~8]{L-vF2}.)
We illustrate this formula in \S\ref{explicit tf} by providing (as well as deriving via a direct computation) the fractal tube formula for the $p$-adic Euler string,
 the definition of which is given in \S\ref{Euler string}. Many further illustrations of our nonarchimedean fractal tube formula are provided in \cite{LapLu3} for the important case of general $p$-adic self-similar strings, including the 3-adic Cantor string (Example \ref{cantor volume} below) and the 2-adic Fibonacci string. 
In \S\ref{towards}, we briefly discuss possible future research directions connected with the theory developed in this paper and its predecessors, \cite{LapLu1, LapLu2, LapLu2.5, LapLu3}. 
Finally, in \S\ref{epilogue}, we reveal the existence of $p$-adic fractal strings of any rational dimension between $0$ and $1$ and a possible connection between their construction and the Riemann hypothesis for the Riemann zeta function. We also discuss some constructions of ad\`elic fractal strings and a geometric zeta function for the ad\`elic Euler--Riemann string.

For more information about the theory of fractal strings (or sprays) and their complex dimensions, beside the books \cite{L-vF1}, \cite{L-vF2}, \cite{LapRaZu1} and \cite{Lap2}, we refer to \cite{Elmr}, \cite{Fal2}, \cite{Ham}, \cite{HeLap, HerLap, HerLap2, HerLap3}, \cite{Kom}, \cite{Lap, Lap3}, \cite{Lap4, Lap5, Lap6, 
Llr, LapLu1, LapLu2, LapLu2.5, LapLu3, LapMa, LapNe, LapPe2, LapPe1, LPW, LapPo, LapPo3}, \cite{LapRaZu1.5, LapRaZu2, LapRaZu3, LapRaZu3.5, LapRaZu4, LapRo1}, \cite{MorSepVi}, \cite{Ol}, \cite{Pe}, as well as the relevant references therein.

\section{Nonarchimedean Fractal Strings}\label{pfs}

\subsection{$p$-adic Numbers}
\label{S:p-adic numbers}

Given a fixed prime number $p$,
 any nonzero rational number $x$ can be written as $x=p^v\cdot a/b$,
 for some integers $a$ and $b$ and a unique exponent $v\in \Z$ such that $p$ does not divide $a$ or $b$.
The {\em $p$-adic absolute value\/} is the function $|\cdot|_p\colon\Q \rightarrow [0,\infty)$ given by $|x|_p=p^{-v}$ and $|0|_p=0$.
It satisfies the {\em strong triangle inequality\/}:
 for every $x,y\in \Q$,
$$
|x+y|_p\leq \max\{|x|_p, |y|_p\}.
$$
Relative to the $p$-adic absolute value,
 $\Q$ does not satisfy the archimedean property because for every $x\in \Q$, $|nx|_p$ will never exceed $|x|_p$ for any $n\in \N$. 
The metric completion of $\Q$ with respect to the $p$-adic absolute value $|\cdot|_p$ is the field of $p$-adic numbers $\Q_p$.
More concretely,
every $p$-adic number $z\in \Q_p$ has a unique $p$-adic expansion
$$
z=a_{v}p^{v} + \cdots + a_0 + a_1p + a_2 p^2 + \cdots,
$$
for some  $v\in \Z$ and  digits  $a_i \in \{0, 1, \dots, p-1\}$ for all $i\geq v$ and $a_v\neq0$.
An important subset of $\Q_p$ is the unit ball,
 $\Z_p=\{x\in \Q_p\colon |x|_p\leq 1\}$,
 which can also be represented as follows:
\[
\Z_p=\left\{ a_0 + a_1p + a_2 p^2 + \cdots\colon a_i \in \{0, 1, \dots, p-1\} \text{ for all } i\geq 0 \right\}.
\]
Using this $p$-adic expansion,
 one sees that
\begin{equation}
\Z_p=\bigcup_{a=0}^{p-1} (a+p\Z_p),
\label{decomposition}
\end{equation}
where $a+p\Z_p=\{y\in \Q_p\colon |y-a|_p\le 1/p\}$.
Thus the $p$-adic ball $\Z_p$ is self-similar to $p$ scaled (by the factor 1/$p$) copies of itself. 
Note that $\Z_p$ is compact and thus complete.
Also,
 $\Q_p$ is a locally compact group,
 and hence admits a unique translation invariant Haar measure $\mu_H,$
normalized so that $\mu_H(\Z_p)=1$.
In particular,
 $ \mu_H(a+p^n\Z_p)=p^{-n}$ for every $n\in \Z$. For general references on $p$-adic analysis, we point out, e.g.,
\cite{Kob, Neu, ParSh,Rob,Sch,Ser}.

 \begin{remark}\label{p-adic analysis}
(a)
The distance $d_p$ defined on $\Q_p$ by $d_p(x,y)=|x-y|_p$ is called an {\em ultrametric,}
 since it satisfies the counterpart of the above strong triangle  inequality: 
\begin{equation}\label{ultrametric}
d_p(x,z) \leq \max\{d_p(x,y), d_p(y,z)\}
\end{equation}
for all $x,y,z \in \Q_p$.
Consequently,
 every triangle in $\Q_p$ is isosceles with the two longer sides having the same length:
\begin{gather}\label{E:isosceles}
\text{If }d_p(x,y)>d_p(y,z)\text{ then }d_p(x,z)=d_p(x,y).
\end{gather}
It follows that the center can be chosen anywhere  within the $p$-adic ball ~$B$.
Moreover,
 given any two balls $B_1$ and $B_2$,
 either they are disjoint or one is entirely contained in the other (i.e.,
 $B_1 \subseteq B_2$ or $B_2 \subseteq B_1$).
These special properties are common to all ultrametric spaces (i.e.,
 all metric spaces for which the ultrametric triangle inequality (\ref{ultrametric}) holds).\smallskip

(b)
By definition,
 $\Z_p$ is the (closed) unit ball of ($\Q_p, d_p$).
Moreover,
 $\Z_p$ has the remarkable property of being a ring (since for all $x,y$ in~$\Z_p$,
by~\eqref{ultrametric} again,
$|x+y|_p\leq \max(|x|_p,|y|_p) \leq 1$,
and $|xy|_p=|x|_p|y|_p\leq 1$). 
This is to be contrasted with the fact that $[-1,1]$,
 the unit ball of $\R$,
 is not stable under addition (although it is obviously stable under multiplication);
see \cite{Har}. 
Finally,
 since translations are homeomorphisms,
 every closed ball $B=B(a,r)$ in $\Q_p$ with center $a$ has a radius $r$ of the form $r=p^n$,
\begin{equation}\label{ball}
B(a,r)=a+p^{-n}\Z_p=\{x\in \Q_p \colon  |x-a|_p\leq r\}
\end{equation}
for some $r\in p^{\Z}$,
 the valuation group of the nonarchimedean field $\Q_p$.
We leave it to the reader to investigate the converse statement according to which every convex subset of $\Q_p$ is a metric ball (i.e.,
 an interval);
 see, e.g.,
\cite{Sch}.\smallskip

(c) ($p$-adic intervals).
In the sequel (as well as in part of the literature on $p$-adic analysis,
 see,
 e.g.,
\cite{Kob}),
 the metric balls $B=a+r\Z_p$ 
(with $a\in \Q_p$ and $r\in p^{\Z}$,
 as in (\ref{ball}) just above),
 are sometimes called the `intervals' of~$\Q_p$.
 Note that they are not connected, in the usual topological sense,
 but that they are `convex',
 in the following sense:
 for each $x,y\in B$ and $\alpha \in \Z_p$,
 we have that $\alpha x + (1-\alpha)y\in B$. (Here and henceforth,
 it is useful to think of $\Z_p\subset \Q_p$ as being the analogue of the unit interval $[0,1]\subset \R$, rather than of $[-1,1]$.)\smallskip

(d) (The archimedean/nonarchimedean dichotomy).
A beautiful and classical theorem of Alexander Ostrowski states that each nontrivial absolute value  on
 the field of rational numbers ~$\Q$,
 is  either  equivalent to 
 the standard archimedean absolute value on $\Q$
 or to
 the nonarchimedean $p$-adic absolute value $|\cdot|_p$ for some prime $p$. (Recall that two absolute values are said to be equivalent if they induce the same topology on~$\Q$; this is the case if and only if one is a power of the other.)
 Therefore, infinitely many completions of $\Q$ (one for each prime $p$) are nonarchimedean and $\R$ is the only  completion of $\Q$ that is archimedean.  
For this reason,
 one sometimes writes~$\R=\Q_{\infty}$ and refers to (the equivalence class of) the absolute value~$|\cdot|$ as the `place at infinity',
 associated with the `prime at infinity'  or the `real prime';
 see~\cite{Har}. (We note that Ostrowski's Theorem is usually expressed in terms of valuations rather than of absolute values.
Accordingly,
 a {\em place\/} of $\Q$ is generally defined as an equivalence class of valuations on $\Q$.) With this notation in mind,
 we see that the field~$\Q_{ \infty}$ is archimedean,
 whereas for any (finite) prime $p$,
 $\Q_p$ is a nonarchimedean field. 
The theory of \mbox{$p$-adic} fractal strings developed in~\cite{LapLu1,LapLu2,LapLu2.5,LapLu3} is aimed,
 initially,
 at finding suitable definitions and obtaining results that parallel those corresponding to the theory of real (or archimedean) fractal strings developed in~\cite{L-vF2},
 for example.
As we will see,
 however,
 although there are many analogies between the archimedean and nonarchimedean theories of fractal strings,
 there are also some notable differences between them;
 see, especially, \cite{LapLu2.5}, along with \cite{LapLu1} and \cite{LapLu3}. 
\end{remark}

\subsection{$p$-adic Fractal Strings}

Let  $ \Omega$ be a bounded open subset of $\Q_p$.  
 Then it can be decomposed into a countable union of disjoint open balls with radius $p^{-n_j}$ centered at $a_j\in \Q_p$, 
\[
 a_j + p^{n_j} \Z_p=B(a_j, p^{-n_j})=\{x\in \Q_p ~|~ |x- a_j|_p \le p^{-n_j}\},
\]
  where $n_j \in \Z$ and $j\in \N^{*}$. 
 (We shall often call a $p$-adic ball an \emph{interval}.  By `ball' here, we mean a metrically closed and hence, topologically open and closed ball.) There may be many different such decompositions since each ball can always be decomposed into smaller disjoint balls~\cite{Kob};
 see Equation (\ref{decomposition}).
 However, there is  a canonical decomposition of $\Omega$ into disjoint balls with respect to a suitable equivalence relation, as we now explain.
  
  \begin{definition}\label{relation}
 Let $U$ be an open subset of $\Q_p$. Given $x,y \in U,$ we write that $x\sim y$ if and only if there is a 
 ball $B\subseteq U$ such that $x, y \in B$.
 \end{definition}
 It is clear from the definition that the relation $\sim$ is reflexive and symmetric. To prove the transitivity, let $x \sim y$ and $y\sim z$. Then there are balls $B_1$ containing $x, y$ and $B_2$ containing $y, z$. Thus $y\in B_1 \cap B_2$; so it follows from the ultrametricity of $\Q_p$ that either 
  $B_1 \subseteq B_2$ or $B_2 \subseteq B_1.$   
  In any case, $x$ and $z$ are contained in the same ball; so $x\sim z$. Hence, the above relation
  $\sim$ is indeed an equivalence relation on the open set $U$.
By a standard argument (and since $\Q$ is dense in~$\Q_p$), one shows that there are at most countably many equivalence classes.

\begin{remark}[Convex components] \label{convex component}
The equivalence classes of $\sim$ can be thought of as the `convex components' of $U$.
They are an appropriate substitute in the present nonarchimedean context for the notion of connected components,
 which is not useful in $\Q_p$ since $\Z_p$ (and hence, every interval) is totally disconnected.
Note that given any $x\in U,$ the equivalence class (i.e., the \emph{convex component}) of $x$ is the largest  ball containing $x$ (or equivalently, centered at $x$) and contained in $U$. 
\end{remark}

\begin{definition}\label{p-adic string}
 A \emph{$p$-adic} (or \emph{nonarchimedean}) fractal string
 $ \mathcal{L}_p$ is a bounded open subset $\Omega$ of $\Q_p$. 
 \end{definition}
 Thus it can be written, relative to the above equivalence relation, canonically as a disjoint union of  intervals or balls: 
 \[ \mathcal{L}_p=\bigcup_{j=1}^{\infty} (a_j + p^{n_j} \Z_p)=\bigcup_{j=1}^{\infty} B(a_j, p^{-n_j}).
\]
 Here, $B(a_j, p^{-n_j})$ is the largest ball centered at $a_j$ and contained in $\Omega$.
We may assume that the lengths (i.e., Haar measure) of the intervals 
$a_j + p^{n_j} \Z_p$ are nonincreasing, by reindexing if necessary.  That is, 
\begin{equation}\label{sequence of lengths}
p^{-n_1}\geq p^{-n_2} \geq p^{-n_3} \geq \cdots >0.
\end{equation}

\medskip
Note that, more generally, a $p$-adic fractal string can be defined as an open subset $\Omega$ of $\Q_p$ such that $\mu_H (\Omega) < \infty.$

\begin{definition}\label{zetaLp}
The\emph{ geometric zeta function} of a $p$-adic fractal string $\mathcal{L}_p$ is defined as
\begin{equation} \label{zeta}
\zeta_{\mathcal{L}_p} (s) := \sum_{j=1}^{\infty} (\mu_H (a_j + p^{n_j} \Z_p))^s
= \sum_{j=1}^{\infty} p^{-n_js} 
\end{equation}
for all $s \in \Com$ with $\Re(s)$ sufficiently large.
 \end{definition}

\begin{remark}
The geometric zeta function $\zeta_{\mathcal{L}_p}$ is well defined since the decomposition of 
${\mathcal{L}_p}$ into the disjoint intervals $a_j + p^{n_j}\Z_p$ is unique. Indeed, these intervals are the equivalence classes of which the open set $\Omega$ (defining $\String_p$) is composed. In other words, they are the $p$-adic  ``convex components'' (rather than the connected components) of $\Omega$. Note that in the real (or archimedean) case, there is no difference between the convex or connected components of $\Omega$,
and hence the above construction would lead to the same sequence of lengths as in~\cite[\S1.2]{L-vF2}.  
\end{remark}

\subsection{Example: $p$-adic Euler String}\label{Euler string}

The following $p$-adic Euler string  is a new example of $p$-adic fractal string,
 which is not self-similar (in the sense of~\cite{LapLu2, LapLu3}).
It is a natural $p$-adic counterpart of  the \emph{elementary prime string},  which is the \emph{local} constituent of the \emph{completed  harmonic string}; cf.~\cite[\S4.2.1]{L-vF2}.    

Let $X=p^{-1}\Z_p$. Then, by the `self-duplication' formula~\eqref{decomposition},
\[ 
 X= \bigcup_{\xi=0}^{p-1} (\xi p^{-1} + \Z_p).
  \]
We now keep the first subinterval $\Z_p$, and then decompose the next subinterval further.
That is,
we write
   \[
p^{-1} + \Z_p= \bigcup_{\xi=0}^{p-1} (p^{-1} + \xi  + p\Z_p).
\]
Again,
 iterating this process, we keep the first subinterval $p^{-1} + p\Z_p$ in the above decomposition and decompose the next subinterval,
   $p^{-1} + 1 + p\Z_p$.
Continuing in this fashion, we obtain an infinite sequence of disjoint subintervals 
   $\left\{ a_n + p^n \Z_p\right\}_{n=0}^{\infty},$
where
   $\left\{a_n\right\}_{n=0}^{\infty}$ satisfies the following initial condition and recurrence relation: 
\[
a_0 = 0 \quad  \mbox{and} \quad  a_n=a_{n-1} + p^{n-2} \quad \mbox{for all} ~n\geq 1.
\]
We call the corresponding $p$-adic fractal string,
   \[
   \mathcal{E}_p=\bigcup_{n=0}^{\infty} (a_n + p^n \Z_p),
   \]
   the  \emph{$p$-adic Euler string.} 
      

The geometric zeta function of the $p$-adic Euler string   $\mathcal{E}_p$  is
\begin{equation*}\label{euler factor}
\zeta_{\mathcal{E}_p} (s) = \sum_{n=0}^{\infty} (\mu_H( a_n + p^n \Z_p)) ^s 
= \sum_{n=0}^{\infty} p^{-ns}=
\frac{1}{1-p^{-s}},  \hspace{1cm} \mbox{for}\,\, \Re(s) > 0.
\end{equation*}

Therefore, $\zeta_{\mathcal{E}_p}$ has a meromorphic extension to all of $\Com$ given by the last expression, which is the classic $p$-\emph{Euler factor} (i.e., the local Euler factor associated with the prime $p$):
\begin{equation}\label{ptheuler}
\zeta_{\mathcal{E}_p} (s) =
\frac{1}{1-p^{-s}},  \hspace{1cm} \mbox{for all } \,s\in \Com.
\end{equation}
 Hence, the set of complex dimensions of $\mathcal{E}_p$  is given by
\begin{equation}\label{cdes}
 \mathcal{D}_{\mathcal{E}_p}=\{ D + i\nu\textbf{p}~ |~ \nu \in \Z\}, 
 \end{equation}
  where $D=\sigma=0$ and $\textbf{p}=2\pi /{\log{p}}.$

\begin{remark}[The punctured unit ball]
The unit ball minus the origin
is not a ball itself,
but instead the infinite union
\[
\Z_p\backslash\{0\}=\bigcup_{n=0}^\infty\bigcup_{k=1}^{p-1}kp^{n}+p^{n+1}\Z_p,
\]
where every time,
a small punctured neighborhood of $0$,
namely $p^n\Z_p\backslash\{0\}$,
is subdivided into smaller balls.
This union is isomorphic to the Euler string:
\[
\mathcal{E}_p=\frac1p\Z_p\backslash\left\{\frac1{p(1-p)}\right\}.
\]
\end{remark}

\begin{remark}[Ad\`elic Euler string] 
Note   that  $\zeta_{\mathcal{E}_p}$ is  the $p$-Euler factor of the Riemann zeta function; i.e., 
\begin{equation*}\label{Riemann}
\prod_{p<\infty} \zeta_{\mathcal{E}_p} (s)
=\prod_{p<\infty} \frac{1}{1-p^{-s}}= \sum_{n=1}^{\infty} \frac{1}{n^s} =\zeta(s) \hspace{1cm}  \mbox{for} \,\,\Re(s) > 1.
\end{equation*}
Recall that the meromorphic continuation $\xi$ of the completed Riemann zeta function has the same (critical) zeros as $\zeta$ and satisfies the functional equation $  \xi (s)= \xi (1-s)$.

We aim to form a certain `ad\`elic product'   over all   $p$-adic Euler strings (including the prime at infinity) so that the geometric zeta function of the resulting ad\`elic Euler string $\mathcal E$ is the \emph{completed} Riemann zeta function.  Formally, the ad\`elic Euler string may be written as 
\[\mathcal{E}=\bigotimes_{p\leq \infty} \mathcal{E}_p\] and its geometric zeta function 
$\zeta_{\mathcal E}(s)$ would then coincide with the completed Riemann zeta function $\xi$
(see~\cite{Rie} and, e.g.,~\cite{Edw}):
  \[ \zeta_{\mathcal E}(s)=\xi(s):=\pi^{-s/2}\Gamma(s/2)\prod_{p<\infty} \frac{1}{1-p^{-s}}.\] 
  \end{remark}
  
\begin{remark}[Comparison with the archimedean theory]\label{harmonic string}
  From the geometric point of view, the nonarchimedean Euler string $\mathcal E_p$ is more natural than its archimedean counterpart, the $p$-elementary prime string $h_p$,
described in~\cite[\S4.2.1]{L-vF2}. 
Indeed, as we have just seen, $\mathcal E_p$ has a very simple geometric definition. Since, by construction, $\mathcal E_p$ and $h_p$ have the same sequence of lengths $\left\{p^{-n}\right\}_{n=0}^{\infty},$
they have the same geometric zeta function, namely, the $p$-Euler factor
\begin{equation}
\zeta_{{h}_p(s)}:=\frac{1}{1-p^{-s}}
\end{equation}
of the Riemann zeta function $\zeta(s),$ and hence, the same set of complex dimensions
\begin{equation}
 \mathcal{D}_{{h}_p}=\left\{i\nu \frac{2\pi}{\log p}\colon \nu\in \Z\right\}.
 \end{equation}  
 An `ad\`elic version' of the `harmonic string' $h$, a generalized fractal string whose geometric zeta function is $\zeta_h(s)=\zeta(s)$,  or rather, of its completion $\tilde{h}$ 
 (so that $\zeta_{\tilde h}(s)=\xi(s)$),
is provided in~\cite[\S4.2.1]{L-vF2}.
In particular,
with each term being interpreted as a positive measure on $(0,\infty)$ and the symbol $\ast$ denoting multiplicative convolution on $(0,\infty)$, we have that 
 \begin{equation}
 h=*_{p<\infty}h_p \quad \mbox{and} \quad  \tilde h=\ast_{p\leq\infty}h_p.
\end{equation}
  
 Furthermore, a noncommutative geometric version of this construction is provided in~\cite{Lap2} in terms of the `prime fractal membrane'; see especially,~\cite[Chaps.~3 and 4]{Lap2},
along with~\cite{LapNe}.
 Heuristically, a `fractal membrane' (as introduced in~\cite{Lap2}) is a kind of ad\`elic, noncommutative torus of infinite genus. It can also be thought of as a `quantized fractal string';
see~\cite[Chap.~3]{Lap2}.
It is rigorously constructed in~\cite{LapNe} using Dirac-type operators, Fock spaces, Toeplitz algebras \cite{BS}, and associated spectral triples (in the sense of~\cite{Con});
see also~\cite[\S4.2]{Lap2}.
   We hope in the future to obtain a suitable nonarchimedean version of that construction. It is possible that in the process, we will establish contact with the physically motivated work in~\cite{Drag} involving 
   $p$-adic quantum mechanics. 
 \end{remark}

\section{The Geometric Zeta Function}\label{GZF}

\begin{figure}[h]
\begin{picture}(324,210)(-180,-105)
\unitlength .85pt
\put(-170,0){\line(1,0){290}}
\put(0,-130){\line(0,1){260}}
\put(70,-2){\line(0,1){4}}\put(70,-8){\makebox(0,0){\small$D$}}
\put(-5,-8){\makebox(0,0){\small$0$}}
\put(100,-2){\line(0,1){4}}\put(100,-8){\makebox(0,0){\small$1$}}
\put(20,70){\small$W$}
\put(-90,58){\small$S$}
\bezier{300}(-90,0)(-90,20)(-80,40)
\bezier{300}(-60,117)(-60,80)(-80,40)
\bezier{300}(-90,0)(-90,-20)(-80,-40)
\bezier{300}(-60,-117)(-60,-80)(-80,-40)
\end{picture}
\caption{The screen $S$ and the window $W$.}
\label{screen window}
\end{figure}
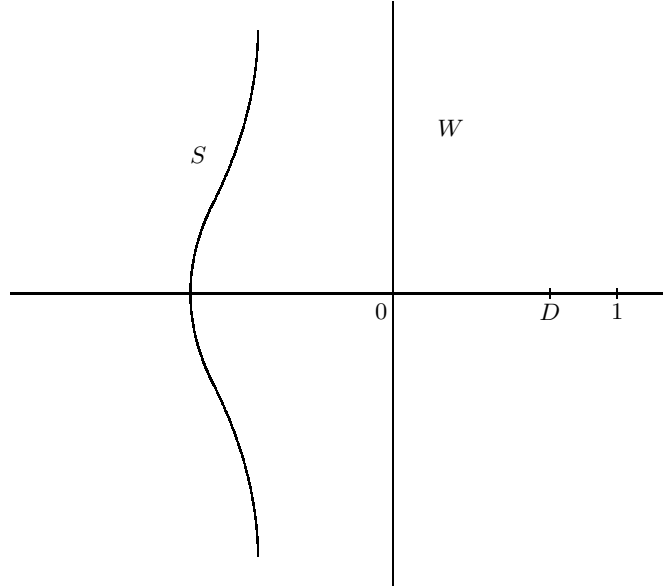
  
The \emph{screen} $S$ is the graph (with the vertical and horizontal axes interchanged)
of a real-valued, bounded and Lipschitz continuous function $S(t)$:
\[
S:=\{S(t) + it ~|~ t\in \R\}.
\]
The \emph{window} $W$ is the part of the complex plane to the right of the screen $S$
(see Figure 1): 
\[
W:=\{s\in \Com ~|~ \Re(s)\geq S(\Im (s))\}.
\]
Let 
\[
\inf S:=\inf_{t\in \R} S(t) \quad \mbox{and} \quad \sup S:=\sup_{t\in \R}S(t), 
\]
and assume that $\sup S \leq \sigma,$ where $\sigma=\sigma_{\String_p}$ is the abscissa of convergence of $\zeta_{\String_p}$ (to be precisely defined in (\ref{sigma}) below).

  \begin{definition}\label{dvcd} 
 {\em Let $\mcl_p$ be a} $p$-adic {\em fractal string.} If $\zeta_{\mathcal{L}_p}$ has a meromorphic continuation to an open connected neighborhood of $W\subseteq \Com$, then 
\begin{equation}\label{vcd}
 \mathcal D_{\String_p}(W):=\{\omega \in W ~|~ \omega \mbox{ is a pole of} ~ \zeta_{\mathcal{L}_p}\} 
 \end{equation}
is called the set of \emph{visible complex dimensions} of $\String_p$.  
If no ambiguity may arise or  if $W=\Com$,  we simply write $\mathcal D_{\String_p}=\mathcal{D}_{\mathcal{L}_p}(W)$ and call it the set of \emph{complex dimensions} of $\String_p$. 

Moreover, the \emph{abscissa of convergence} of $\zeta_{\String_p}$ (where $\mcl_p$ is defined in Equation~\eqref{zeta}) is denoted by
$\sigma=\sigma_{\String_p}$.
Recall that it is defined by (see, e.g.,~\cite{Ser}) 
\begin{equation}\label{sigma}
\sigma_{\String_p}:=\inf\Bigg\{\alpha \in \R \colon \sum_{j=1}^{\infty}p^{-n_j\alpha} <\infty\Bigg\}.
\end{equation}
\end{definition}

\begin{remark}
In particular, if $\zeta_{\String_p}$ is entire (which occurs only in the trivial case when 
$\String_p$ is given by a finite union of intervals), then $\sigma_{\String_p}=-\infty.$
Otherwise,
 $\sigma_{\String_p}\geq 0$ (since $\String_p$ is composed of infinitely many intervals) and we will see in Theorem~\ref{T:p-adic D_M=sigma} that $\sigma_{\String_p}<\infty$ since $\sigma_{\String_p}\leq D_M\leq 1,$
 where $D_M=D_{M,\String_p}$ is the Minkowski dimension of $\String_p$,
to be introduced in \S\ref{Mdimension}.
Furthermore,
it will follow from Theorem~\ref{T:p-adic D_M=sigma} that for a nontrivial $p$-adic fractal  string, 
$\sigma_{\String_p}=D_M.$ 

Observe that since $ \mathcal D_{\String_p}(W)$ is defined as a subset of the poles of a meromorphic function, it is at most countable and forms a discrete subset of $\Com$. 

Finally, we note that it is well known that $\zeta_{\String_p}$ is holomorphic for $\Re(s)>\sigma_{\String_p};$ see, e.g.,~\cite{Ser}.
Hence, 
\[
\mathcal D_{\String_p} \subseteq \{s\in \Com\colon \Re(s)\leq \sigma_{\String_p}\}.
\]
\end{remark}

\begin{remark}[Archimedean fractal strings]\label{afs}
\emph{Archimedean} or \emph{real} fractal strings are defined as bounded open subsets of the real line $\R=\Q_{\infty}.$ They were initially defined in~\cite{LapPo},
following an early example in~\cite{Lap1},
and have been used extensively in a variety of settings; see, e.g.,~\cite{Elmr,Fal2,HeLap,HerLap,HerLap2,HerLap3,Kom,Lap1,Lap,Lap3,Lap4,Lap5,Llr,LapMa,LapPe2, LapPe1,LPW,LapPo3,LapRaZu1,LapRaZu1.5,LapRaZu2,LapRaZu3,LapRaZu3.5,LapRaZu4, LapRo1,MorSepVi,Pe} and the books~\cite{L-vF1,LapRaZu1,L-vF2,Lap2}.
Since an open set $\Omega\subset \R$ is canonically equal to the disjoint union of finitely or countably many open and bounded intervals (namely, its connected components), say $\Omega=\bigcup_{j=1}^{\infty} I_j,$ we may also describe a real fractal string by a sequence of lengths $\String=\left\{l_j\right\}_{j=1}^{\infty},$ where $l_j=\mu_L(I_j)$ is the length or 1-dimensional Lebesgue measure of the interval $I_j$, written in nonincreasing order:
\[
l_1\geq l_2\geq l_3 \geq \cdots.
\]
(A justification for this identification is provided by the formula for the volume $V_{\String}(\eps)$ of $\eps$-inner tubes of $\Omega$, as given by Equation~\eqref{LPformula} below.) 
Note that since $\mu_L(\Omega) < \infty,  l_j \rightarrow 0$ as $j\rightarrow \infty$, except in the trivial case when $\Omega$ consists of finitely many intervals. 
Also observe that the 1-dimensional Lebesgue measure $\mu_L$ is nothing but the Haar measure on $\R=\Q_{\infty}$, normalized so that $\mu_L([0,1])=1.$ 

All of the definitions given above for $p$-adic fractal strings have a natural counterpart for real fractal strings.
For instance, the geometric zeta function of $\String$ is initially defined by 
\begin{equation}\label{zetaL}
\zeta_{\String}(s)=\sum_{j=1}^{\infty}(\mu_L(I_j))^s=\sum_{j=1}^{\infty} l_j^s,
\end{equation}
for $\Re(s) > \sigma_{\String},$
 the abscissa of convergence of $\zeta_\String$,
 and for a given screen $S$  and associated window $W$,
 the set 
$\mathcal D_{\String}=\mathcal D_{\String}(W)$ of visible complex dimensions of~$\String$ is given exactly as in (\ref{vcd}) of Definition \ref{dvcd},
 except with $\String_p$ and $\zeta_{\String_p}$
replaced with~$\String$ and $\zeta_{\String}$,
respectively.
Similarly, $\sigma_{\String}$, the abscissa of convergence of $\zeta_\String$ is given as in~\eqref{sigma},
 except with the lengths of $\String$ instead of those of~$\String_p$.
Moreover,
it follows from~\cite[Thm.~1.10]{L-vF2} that for any nontrivial real fractal string $\String$, we have $\sigma_{\String}=D_M$, the Minkowski dimension of $\String$ (i.e., of its topological boundary $\partial \Omega$). This latter result will be given a new proof in \S 5.1.

We refer the interested reader to the research monographs~\cite{L-vF1,L-vF2} for a full development of the theory of real fractal strings and their complex dimensions.
\end{remark}

\subsection{Languid and Strongly Languid $p$-adic Fractal Strings}

In \S\ref{explicit tf}, we will obtain  explicit tube formulas for $p$-adic fractal strings, with and without error term.
(See Theorem \ref{dtf} and Corollary \ref{ftf}.)
We will then apply the  tube formula without error term (the strongly languid case of Theorem \ref{dtf}) to the $p$-adic Euler string discussed in \S\ref{Euler string} and revisited in Example~\ref{eulerrevisit} (at the end of \S\ref{explicit tf}). 

In order to state the explicit formulas with (or without) error term, we need to assume the following technical hypotheses (see~\cite[Defns.~5.2 and~5.3]{L-vF2} and recall the definition of the screen $S$ given in \S\ref{pfs},
just before Definition \ref{dvcd}).

\begin{definition}\label{languid}
A $p$-adic  fractal string $\mathcal{L}_p$ is said to be \emph{languid} if its geometric zeta function 
$\zeta_{\mathcal{L}_p}$ satisfies the following growth conditions: There exist real constants $\kappa$ and $C>0$ and a two-sided sequence $\{T_n\}_{n\in \Z}$ of real numbers such that $T_{-n}<0<T_n$ for $n\geq1$, and 
\[
\lim_{n \rightarrow \infty} T_n=\infty, \quad  \lim_{n \rightarrow \infty} T_{-n}=-\infty, \quad
 \lim_{n \rightarrow \infty} \frac{T_n}{|T_{-n}|}=1,
 \]
 such that
 \begin{itemize}
 \item $\mathbf{L1}$ For all $n\in \Z$ and all $u \geq S(T_n),$
\[| \zeta_{\mathcal{L}_p}(u +i T_n)|\leq C (|T_n|+1)^{\kappa},\]
  \item $\mathbf{L2}$
 For all $t\in \R, |t|\geq 1,$
\[ | \zeta_{\mathcal{L}_p}(S(t)+it)|\leq C |t|^{\kappa}.\]
 \end{itemize}

We say that $\mathcal{L}_p$ is \emph{strongly languid} if its geometric zeta function $\zeta_{\mathcal{L}_p}$ satisfies the following conditions, in addition to $\mathbf{L1}$ with $S(t)\equiv -\infty:$
 There exists a sequence of screens $S_m: t\mapsto S_m(t)$ for $m\geq1, t\in \R$, with 
 $\sup S_m \rightarrow -\infty$ as $m\rightarrow \infty$ and with a uniform Lipschitz bound $\sup_{m\geq1} ||S_m||_{Lip}< \infty$, such that
 \begin{itemize}
 \item $\mathbf{L2'}$ There exist constants $A,C>0$ such that for all $t\in \R$ and $m\geq1$, 
 \[ |\zeta_{\mathcal{L}_p}(S_m(t)+it)|\leq C A^{|S_m(t)|}(|t| + 1)^{\kappa}.\]
 \end{itemize}
 \end{definition}
 
\begin{remark} (a) Intuitively, hypothesis $\mathbf{L1}$ is a polynomial growth condition along horizontal lines (necessarily avoiding the poles of $\zeta_{\mathcal{L}_p}$), while hypothesis $\mathbf{L2}$ is a polynomial growth condition along the vertical direction of the screen.

(b) Clearly, condition $\mathbf{L2'}$ is stronger than $\mathbf{L2}$. Indeed, if $\mathcal{L}_p$ is strongly languid, then it is also languid (for each screen $S_m$ separately). 

(c) Moreover, if $\mathcal{L}_p$ is languid for some $\kappa$, then it is also languid for every larger value of $\kappa$.  The same is also true for strongly languid strings. 

(d) Finally, hypotheses $\mathbf{L1}$ and $\mathbf{L2}$ require that  $\zeta_{\mathcal{L}_p}$ has an analytic (i.e., meromorphic) continuation to an open,
 connected neighborhood of $\Re(s)\geq \sigma_{\String_p}$, while $\mathbf{L2'}$ requires that $\zeta_{\String_p}$ has a meromorphic continuation to all of $\Com$. 
\end{remark}

\section{Volume of Thin Inner Tubes}\label{inner tube}

In this section, we provide a suitable analog in the $p$-adic case of the `boundary' of a fractal string and of the associated inner tubes 
 (or  ``inner $\eps$-neighborhoods''). Moreover, we give the $p$-adic counterpart of the expression that yields the volume of the inner tubes (see Theorem \ref{thin}). This result will serve as a starting point in \S6 for proving the distributional explicit tube formula obtained in Theorem \ref{dtf}. 

\begin{definition}\label{volume definition}
Given a point $a\in \Q_p$ and a positive real number $r>0$,  
let $B=B(a, r) =\{x\in \Q_p ~|~ |x-a|_p \le r \}$ be a \emph{metrically closed} ball in $\Q_p,$ as above. (Recall that it follows from the ultrametricity of $|\cdot|_p$ that $B$ is topologically both closed and open (i.e., clopen) in $\Q_p$.)
 We  call 
$S=S(a, r)=\{x\in \Q_p ~|~ |x-a|_p = r \}$ the \emph{sphere} of $B$. (In our sense, $S$ also coincides with the `metric boundary' of $B$, as given in the next definition.)

Let $\mathcal{L}_p= \bigcup_{j=1}^{\infty} B(a_j, r_j)$ be a $p$-adic fractal string. We then define the \emph{metric boundary} $\beta\mathcal{L}_p$ of $\mathcal{L}_p$ to be the disjoint union of the corresponding spheres, i.e.,  
\[
\beta\mathcal{L}_p = \bigcup_{j=1}^{\infty} S(a_j, r_j).
\]

Given a real number $\eps>0$,  define the  \emph{thick $p$-adic `inner 
$\eps$-neighborhood'} (or `\emph{inner tube}') of $\mathcal{L}_p$ to be
\begin{equation}\label{thick inner tube}
\mathcal{N}_{\eps}=\mathcal{N}_{\eps}(\String_p):=\{ x\in \mathcal{L}_p ~|~  d_p(x, \beta\mathcal{L}_p) < \eps\},
\end{equation}
where $d_p(x, E)=\inf \{ |x-y|_p ~|~  y\in E\}$
is the $p$-adic distance of $x \in \Q_p$ to a subset $E \subseteq \Q_p$.
Then the \emph{volume $\mathcal{V}_{\String_p}(\eps)$ of the thick inner 
$\eps$-neighborhood} of $\mathcal{L}_p$  is defined to be the Haar measure of $\mathcal{N}_{\eps}$, i.e., 
$\mathcal{V}_{\String_p}(\eps)=\mu_H(\mathcal{N}_{\eps}).$
\end{definition}
 
\begin{lemma}\label{ball and sphere}
Let $B=B(a,r)$ and $S=S(a,r)$, as in Definition \ref{volume definition}. Then, for any positive number $\eps <r$, we have 
\begin{equation}\label{sphere}
\mathcal{N}_{\eps}(B):=\{x\in B ~|~ d_p(x, S)<\eps\}=S. 
\end{equation}
Further,
 if $r=p^{-m}$ for some $m\in \Z,$ then for all $\eps <r,$
\begin{equation}\label{spherevolume}
\mu_H(\{x\in B ~|~ d_p(x, S)<\eps\})=\mu_H(S)=
(1-p^{-1})p^{-m}.
\end{equation}
\end{lemma}

\begin{proof}
 Clearly $S \subseteq\{x\in B ~|~ d_p(x, S)<\eps\}$   since for any $x\in S$, 
$d_p(x,S)=0$. Next, fix $\eps$ with $0<\eps<r$ and let $x\in B$ be such that $d_p(x, S)< \eps.$ Then there must exist $y\in S$ such that $|x-y|_p<\eps.$
But, since $|y-a|_p=r,$ we deduce from the fact that every ``triangle'' in $\Q_p$ is isosceles~\cite[p.~6]{Kob} that 
$|x-a|_p=|y-a|_p$ and thus $x\in S$. This completes the proof of~\eqref{sphere}.

We next establish formula~\eqref{spherevolume}.
In light of Equation~\eqref{sphere},
 it suffices to show that 
\begin{equation}\label{Svolume}
\mu_H(S)=(1-p^{-1})p^{-m}.
\end{equation}
Let $S^1=S(0,1)=\{x\in \Q_p~|~ |x|_p=1\}$ denote the unit sphere in $\Q_p$. Since 
$S=S(a, p^{-m})=a+p^mS^1,$ we have that $\mu_H(S)=\mu_H(S^1)p^{-m}.$
Next we note that 
\[
B(0,1)= \bigcup_{m\geq 0}S(0,p^{-m})
\]
is a disjoint union. Hence,  by taking the Haar measure of $B(0,1),$ we deduce that 
\begin{equation}
  1 = \left(\sum _{m=0} ^ {\infty} p^{-m}\right) \mu_{H}(S^1) =  \frac{1}{1-p^{-1}} \mu_{H}(S^1),
  \end{equation}
from which~\eqref{Svolume} and hence,
 in light of the first part,
\eqref{spherevolume} follows. 
\end{proof}

\begin{theorem}[Volume of thick inner tubes]\label{thick}
Let $ \mathcal{L}_p=\bigcup_{j=1}^{\infty}  B(a_j, p^{-n_j})$ be a $p$-adic fractal string.
Then,
for any $\eps >0,$ we have
\begin{align}
\mathcal{V}_{\String_p}(\eps)
&=(1-p^{-1})\sum_{j=1}^k p^{-n_j} +\sum_{j>k} p^{-n_j}\label{equation 1}\\
&= \zeta_{\String_p}(1)-\frac{1}{p} \sum_{j=1}^k   p^{-n_j},\label{equation 2}
\end{align} 
where $k=k(\eps)$ is the largest integer such that $p^{-n_k}\ge \eps$.
\end{theorem}

\begin{proof}
In light of the definition of $\mathcal N_{\eps}=\mathcal N_{\eps}(\String_p)$ given in Equation (\ref{thick inner tube}) and the definition of $k$ given in the theorem, we have that 
\[
\mathcal N_{\eps}=\bigcup_{j=1}^kS_j\cup \bigcup_{j>k}B_j,
\]
where $B_j:=B(a_j,p^{-n_j})$ and $S_j:=S(a_j,p^{-n_j})$ for each $j\geq 1$.

We then apply Lemma \ref{ball and sphere} to deduce the expression of 
$\mathcal V_{\String_p}(\eps)=\mu_H(\mathcal N_{\eps})$
stated in Equations~\eqref{equation 1} and~\eqref{equation 2}.
\end{proof}

Note that
$\zeta_{\String_p}(1)=\sum_{j=1}^{\infty}p^{-n_j}$ 
is the volume of $\String_p$ (or rather, of the bounded open subset $\Omega$ of $\Q_p$ representing $\String_p$):
\[
\zeta_{\String_p}(1)=\mu_H(\String_p)<\infty.
\]
It is clearly independent of the choice of $\Omega$ representing $\String_p$, and so is $\mathcal V_{\String_p}(\eps)$ in light of either~\eqref{equation 1} or~\eqref{equation 2}.

 \begin{corollary}\label{limit}
 The following limit exists in $(0,\infty)\colon$
 \begin{equation}
 \lim_{\eps \to 0^+} \mathcal{V}_{\String_p}(\eps)
 =\mu_H(\beta\String_p)
 =(1-p^{-1})\zeta_{\String_p}(1).
   \end{equation} 
 \end{corollary} 
 
 \begin{proof}
 This follows by letting $\eps \rightarrow 0^+$ in either~\eqref{equation 1} and~\eqref{equation 2} and noting that $k=k(\eps)\rightarrow \infty.$
 \end{proof}
 
Corollary \ref{limit}, combined with the fact that $\beta\String_p \subset \mathcal N_{\eps}(\String_p)$ for any $\eps >0,$ naturally  leads us to introduce the following definition. 

\begin{definition}\label{thin inner tube}
Given $\eps >0,$ the \emph{thin $p$-adic `inner $\eps$-neighborhood'} (or \emph{ `inner tube'}) of 
$\String_p$ is given by 
\begin{equation}\label{equation N}
N_{\eps}=N_{\eps}(\String_p):=\mathcal N_{\eps}(\String_p)\backslash \beta\String_p.
\end{equation}

Then,
in light of Corollary~\ref{limit},
the \emph{volume $V_{\String_p}(\eps)$ of the thin inner  
$\eps$-neigh\-bor\-hood} of $\String_p$ is defined to be the Haar measure of $N_{\eps}$ and is given by
\begin{equation}\label{volume of the thin inner tube}
V_{\String_p}(\eps):=\mu_H(N_{\eps})
=\mathcal V_{\String_p}(\eps)-\mu_H(\beta \String_p).
\end{equation}
Note that, by construction, we now have $\lim_{\eps \to 0^+}V_{\String_p}(\eps)=0.$
\end{definition}

We next state the counterpart (for thin inner tubes) of Theorem \ref{thick}, which is the key result that will enable us to obtain an appropriate $p$-adic analog of the fractal tube formula (in \S6) as well as of the notions of Minkowski dimension and content (in \S\ref{Mdimension}).

\begin{theorem}[Volume of thin inner tubes]\label{thin}
Let $ \mathcal{L}_p=\bigcup_{j=1}^{\infty}  B(a_j, p^{-n_j})$ be a $p$-adic fractal string. Then, for any $\eps >0,$ we have
\begin{align}\label{volumeequation}
V_{\String_p}(\eps)
&=p^{-1}\sum_{j>k} p^{-n_j} =p^{-1}\sum_{j:p^{-n_j}<\eps} p^{-n_j}\\
&= p^{-1}\biggl(\zeta_{\String_p}(1)-\sum_{j=1}^k   p^{-n_j}\biggr),\label{equation2}
\end{align} 
  where $k=k(\eps)$ is the largest   integer such that
 $p^{-n_k}\ge \eps$, as before.
\end{theorem}

\begin{proof}
In view of Theorem \ref{thick} and Corollary \ref{limit}, the result follows immediately from Equation \eqref{volume of the thin inner tube} in Definition \ref{thin inner tube}.
\end{proof}

\begin{remark}
Observe that because the center $a$ of a $p$-adic ball $B=B(a,p^{-n})$ can be chosen arbitrarily without changing its radius $p^{-n},$ the metric boundary of a ball, $\beta B=S=S(a,p^{-n}),$
 depends on the choice of $a$. 
Note, however, that in view of 
Equation~\eqref{spherevolume} in Lemma \ref{ball and sphere}, its volume $\mu_H(S)$ depends only on the radius of $B$.
Similarly,
 even though the decomposition of a $p$-adic fractal string $\Omega$ (i.e., $\String_p$) into maximal balls $B_j=B_j(a_j, p^{-n_j})$ is canonical,
`the' metric boundary of $\String_p,$
 $\beta\mathcal{L}_p = \bigcup_{j=1}^{\infty} S(a_j, r_j),$
 depends on the choice of the centers~$a_j$.
However,
according to Corollary \ref{limit},
$\mu_H(\beta\String_p)$ is independent of this choice and hence, neither 
 $\mathcal V_{\String_p}(\eps)=\mu_H(\mathcal N_{\eps}(\String_p))$ nor 
 $V_{\String_p}(\eps)=\mu_H(N_{\eps}(\String_p))$ 
 depends on the choice of the centers. 
 Indeed, in light of Theorem \ref{thick} and Theorem \ref{thin},~$\mathcal V_{\String_p}(\eps)$ and~$V_{\String_p}(\eps)$
 depend only on the choice of the $p$-adic lengths $p^{-n_j},$ and hence solely on the $p$-adic fractal string $\String_p$, viewed as a nonincreasing sequence of positive numbers, and not on the geometric representation $\Omega$ of $\String_p$, let alone on the choice of the centers of the balls of which $\Omega$ is composed.

Although it is not entirely analogous to it,
 this situation is somewhat  reminiscent of the fact that the volume $V_{\String}(\eps)$ of the inner $\eps$-neighborhoods of an archimedean fractal string depends only on its
lengths~$\left\{l_j\right\}_{j=1}^{\infty}$ and not on the representative 
 $\Omega$ of $\String$ as a bounded open set; see Equation~\eqref{LPformula} and the discussion surrounding it in Remark \ref{real vs. p-adic volume}.
\end{remark}

\begin{remark}[Comparison between the archimedean and the nonarchimedean cases]\label{real vs. p-adic volume}
Recall that $\mathcal V_{\String_p}(\eps)$ does not tend to zero as $\eps \to 0^+,$
 but that instead it tends to the positive number $(1-p^{-1})\zeta_{\String_p}(1),$
 whereas $V_{\String_p}(\eps)$ does tend to zero.
This is the reason why the Minkowski dimension must be defined in terms of $V_{\String_p}(\eps)$ (as will be done in \S \ref{Mdimension}) rather than in terms 
          of $\mathcal V_{\String_p}(\eps).$
                    Indeed, if $\mathcal V_{\String_p}(\eps)$ were used instead, then every $p$-adic fractal string would have Minkowski dimension 1. This would be the case even for a trivial $p$-adic fractal string composed of a single interval, for example. This is also why, in the $p$-adic case, we will focus only on the tube formula for $V_{\String_p}(\eps)$
                    rather than  for   $\mathcal V_{\String_p}(\eps)$, although the latter could be obtained by means of the same techniques.

Note the difference between the expressions for $V_{\String}(\eps)$ in the case of an archimedean fractal string $\String$ and for its nonarchimedean thin (resp., thick) counterpart 
    $V_{\String_p}(\eps)$ (resp., $\mathcal V_{\String_p}(\eps)$)
    in the case of a $p$-adic fractal string $\String_p$. 
    Compare Equation (8.1) of~\cite{L-vF2} (which was first obtained in~\cite{LapPo}),
    \begin{equation}\label{LPformula}
    V_{\String}(\eps)=\sum_{j: l_j\geq 2\eps}2\eps +\sum_{j: l_j< 2\eps}l_j,
    \end{equation}
    with Equations~\eqref{volumeequation}--\eqref{equation2} in Theorem \ref{thin}.
(Here, we are using the notation of Remark \ref{afs}, to which the reader is referred to for a brief introduction to real fractal strings.) 
          It follows, in particular, that   $V_{\String}(\eps)$ is a continuous function of $\eps$ on $(0,\infty)$, whereas  $\mathcal V_{\String_p}(\eps)$ (and hence also $V_{\String_p}(\eps)$) is discontinuous (because it is a step function with jump discontinuities at each point $p^{-n_j}, \mbox{for}~ j=1, 2, 3, \ldots $).    
           The above discrepancies between the archimedean and the nonarchimedean cases help explain why the tube formula for real and $p$-adic fractal strings have a similar form, but with different expressions for the corresponding `tubular zeta function' (in the sense of~\cite{LapPe1,LPW}). We note that a minor aspect of these discrepancies is that $2\eps$ is now replaced by $\eps.$ Interestingly, this is due to the fact that the unit interval $[0,1]$ has inradius 1/2 in $\R=\Q_{\infty}$ whereas 
          $\Z_p$ has inradius 1 in $\Q_p.$ Recall that the \emph{inradius} of a subset $E$ of a metric space is the supremum of the radii of the balls entirely contained in $E$.
          
          Finally, we note that for an archimedean fractal string $\String$, there is no reason to distinguish between the `thin volume' $V_{\String}$ and the `thick volume' 
          $\mathcal V_{\String}$, as we now explain. Indeed, the archimedean analogue $\beta \String$ of the metric boundary is a countable set, and hence has measure zero, no matter which geometric realization $\Omega$ one chooses for $\String$.
          More specifically, in the notation of Remark \ref{afs}, $\beta \String$ consists of all the endpoints of the open intervals $I_j$ (the connected components of $\Omega$, or equivalently, its convex components). Hence, $\mu_L(\beta \String)=0$ and so 
\[
V_{\String}(\eps):=\mathcal V_{\String}(\eps)-\mu_L(\beta \String)=\mathcal V_{\String}(\eps),
\]
as claimed.

For example,
if $\String$ is the ternary Cantor string $\mathcal {CS}$,
then $\beta \String$ is the countable set consisting of all the endpoints of the `deleted intervals'
in the construction of the real Cantor set $\mathcal {C}$.
In other words,
$\beta\String$ is the set $\mathcal T$ of ternary points (which has measure zero because it is countable).
  Hence, the metric boundary $\beta \String$ of $\mathcal{CS}$  is dense in 
  $\partial \String$, the topological boundary of $\mathcal {CS}$, and which in the present case, coincides with the ternary Cantor set $\mathcal C$. Also note that the fact that 
  $\mathcal C=\partial \String$ (and not $\mathcal T=\beta\String$) has measure zero is purely coincidental and completely irrelevant here. 
          Indeed, the same type of argument would apply if $\String$ were any archimedean fractal string, even if $\mu_L(\partial \String)>0$ as is the case for example, if $\partial\String$ is a `fat Cantor set' (i.e., a Cantor set of positive measure) or,
 more generally,
 if $\partial\String$ is a `fat fractal'
 (in the sense of~\cite{GMOY,Ott}).
          The underlying reason is that in the archimedean case, the topological boundary $\partial \String=\partial \Omega$ is disjoint from $\Omega$  (since $\Omega$ is open), and hence, does not play any role in the computation of $V_{\String}(\eps)$ or of $\mathcal V_{\String}(\eps)$. 
          By contrast, it is not  true that the metric boundary  $\beta \String$ and the geometric representation $\Omega$ are disjoint (since, in fact,  $\beta \String \subseteq \Omega$), but what is remarkable is that the Minkowski dimension of $\beta \String$ coincides with that of its closure, and hence (in most cases of interest),
           with $D_{M,\String}$.
\end{remark}

\subsection{Example: The Euler String}
\label{euler volume}

As a first application of Theorem \ref{thin}, we can obtain, via a direct computation, a tube formula for the $p$-adic Euler string $\mathcal E_p$; that is, an explicit formula for the volume of the thin inner 
 $\eps$-neighborhood, $V_{\mathcal E_p}(\eps)$, as given in Definition \ref{thin inner tube}. 
     
Let $\mathcal E_p$ be the $p$-adic Euler string defined in \S\ref{Euler string}.  Given  
  $\eps >0$, let $k$ be the largest integer such that $\mu_H(a_k+p^k\Z_p)=p^{-k}\geq \eps$;
 then $k=[\log_p \eps^{-1}]$. (Here, for $x\in \R$, we write $x=[x]+\{x\},$ where $[x]$ is the integer part  and $\{x\}$ is the fractional part of 
$x$; i.e., $[x]\in \Z$ and $0\leq \{x\}<1$.) Thus,
by Equation (\ref{volumeequation}) of Theorem \ref{thin},
we have, 
\[
V_{\mathcal E_p}(\eps)
=p^{-1}\sum_{n=k+1}^{\infty} p^{-n}
=\frac{p^{-1}}{p-1}p^{-k}
= \frac{p^{-1}}{p-1}p^{-\log_p \eps^{-1} }
  \left(\frac{1}{p}\right)^{-\{\log_p\eps^{-1}\}}\label{V(epsilon)},
\]
since $k=\log_p \eps^{-1} -\{\log_p \eps^{-1} \}$.
Next,
the Fourier series expansion for $b^{-\{x\}}$ is given by
(see~\cite[Eq.~(1.13)]{L-vF2})
    \begin{equation}\label{fourier}
    b^{-\{x\}}=\frac{b-1}{b}\sum_{n\in \Z}\frac{e^ {2\pi i n x}}{\log b+2\pi in},
    \end{equation}
Applying it with $b=1/p$ and $x=\log_p\eps ^{-1}$,
we find
\begin{align}
V_{\mathcal E_p}(\eps)
&= \frac{p^{-1}}{p-1}  \frac{p-1}{\log p}\sum_{n\in \Z}\frac{\eps^{1 -in\textbf p}}{1 -in\textbf p}\nonumber  \\
  &= \frac{1}{p\log p}\sum_{\omega \in \mathcal D_{\mathcal E_p} } 
     \frac{\eps^{1-\omega}}{ 1-\omega }.\label{VolumeEuler}
\end{align}
    Finally, in the last equality, we have used Equation (\ref{cdes}) for the set of complex dimensions $\mathcal D_{\mathcal E_p}$ of $\mathcal E_p$. 

\section{Minkowski Dimension}
\label{Mdimension}

In the sequel, the (inner) Minkowski dimension and the (inner) Minkowski content of a $p$-adic fractal string $\String_p$ (or, equivalently, of its metric boundary $\beta \String_p$, see Definition
  \ref{volume definition}) is defined exactly as the corresponding notion for a real fractal string 
  (see~\cite[Defn.~1.2]{L-vF2}), 
 except for the fact that we now use the definition of $V(\eps)=V_{\String_p}(\eps)$ provided in Equation (\ref{volume of the thin inner tube}) of Definition \ref{thin inner tube}.
(For reasons that will be clear to the reader later on in this section,
we denote by $D_M=D_{M,\String_p}$ instead of by $D=D_{\String_p}$ the Minkowski dimension of $\String_p.$) More specifically,
 the \emph{Minkowski dimension} of $\String_p$ is given by 
\begin{equation}\label{dimension}
D_M= D_{M,\String_p}:= \inf  \left\{\alpha \geq  0~|~ V_{\String_p}(\eps)=O(\eps ^{1-\alpha}) ~\mbox{as}~
\eps \rightarrow 0^{+}  \right\}.  \end{equation}
Furthermore, $\String_p$ is said to be \emph{Minkowski measurable}, with \emph{Minkowski content}
$\mathcal M$, if the limit 
\begin{equation}
\mathcal M=\lim_{\eps\rightarrow 0^{+}} V_{\String_p}(\eps)\eps ^{-(1-D_M)}
\end{equation}
exists in  $(0, \infty).$

\begin{remark}
Note that since 
$V_{\String_p}(\eps)=\mathcal V_{\String_p}(\eps)-\mu_H(\beta \String_p)$,
the above definition of the Minkowski dimension is somewhat analogous to that of ``exterior dimension'', which is sometimes used in the archimedean case to measure the roughness of a `fat fractal' (i.e., a fractal with positive Lebesgue measure). The notion of exterior dimension has been useful in the study of aspects of chaotic nonlinear dynamics;
see, e.g.,~\cite{GMOY} and the survey article~\cite{Ott}.
\end{remark}

The goal of the rest of this section is to establish the following theorem,
which is the exact analogue for $p$-adic fractal strings of~\cite[Thm.~1.10]{L-vF2}, which was first observed in \cite{Lap, Lap3} by using a result of \cite{BT}. (Recall that $\sigma_{\mathcal{L}_p}$ is defined in Equation \eqref{sigma} of \S\ref{GZF}. Also note that we need to assume that $\String_p$ has infinitely many lengths since if $\String_p$ is composed of finitely many intervals,
then $\sigma_{\String_p}=-\infty$ and $D_M=D=0$; see formula \eqref{E:dimension} below for the definition of $D$.\label{F: j->infty})

\begin{theorem}\label{T:p-adic D_M=sigma}
Let\/ $\String_p$ be a $p$-adic fractal string composed of infinitely many intervals.
Then the Minkowski dimension\/ $D_M=D_{M,\String_p}$ of\/ $\String_p$ equals the abscissa of convergence\/ $\sigma_{\String_p}$ of the geometric zeta function\/ $\zeta_{\String_p}$.
That is,
$D_M=\sigma_{\String_p}$.
\end{theorem}

Theorem~\ref{T:p-adic D_M=sigma} will be established in Theorem~\ref{T:p-adic D=sigma=D_M} below in greater generality,
namely for any summable sequence of positive numbers $l_j$.
This is the object of the technical Lemma~\ref{L:zeta=N,V},
which is of independent interest.

When applied to a $p$-adic fractal string,
the following lemma relates the thin volume with the zeta function.
For completeness,
but independently of this,
we also formulate the counterpart for the counting function of the reciprocal lengths of an arbitrary fractal string.
The lemma holds in general,
independently of the fact that in the present situation,
the lengths are powers of $p$.
Recall from~\eqref{volumeequation} that
$$
V(\eps)=\frac1p\sum_{j\colon l_j\leq\eps}l_j,
$$
writing $V$ instead of $V_{\String_p}$ since what follows holds for arbitrary infinite sequences of positive numbers~$l_j$ such that $\sum_{j=1}^\infty l_j$ is convergent.
Also, the {\em geometric counting function} of $\mathcal{L} := \{l_j\}_{j=1}^\infty$,
\begin{gather}\label{E:Nx}
N(x):=\sum_{l_j\geq1/x}1,
\end{gather}
 is the number of reciprocal lengths up to $x > 0$,
and
\begin{gather}\label{E:zetaL}
\zeta_\String(s):=\sum_{j=1}^\infty l_j^s,
\end{gather}
at least for all $s \in \mathbb{C}$ such that $\mfr s > 1$.

\begin{lemma}\label{L:zeta=N,V}
We have the following two expressions for $\zeta_\String(s):$
\begin{gather}\label{E:zeta=V}
\zeta_\String(s)=\zeta_\String(1)l_1^{s-1}+(1-s)\int_0^{l_1}pV(\eps)\eps^{s-2}\,d\eps,
\end{gather}
and
\begin{gather}\label{E:zeta=N}
\zeta_\String(s)=s\int_0^\infty N(x)x^{-s-1}\,dx.
\end{gather}
Both expressions converge exactly when $\sum_{j=1}^\infty l_j^s$ converges.
\end{lemma}

\begin{proof}
For $n>0$,
we compute
\begin{align*}
(1-s)\int_{l_n}^{l_1}pV(\eps)\eps^{s-2}\,d\eps
&=\sum_{j=1}^{n-1}(1-s)\int_{l_{j+1}}^{l_j}pV(\eps)\eps^{s-2}\,d\eps\\
&=\sum_{j=1}^{n-1}\sum_{k=j+1}^\infty l_k(1-s)\int_{l_{j+1}}^{l_j}\eps^{s-2}\,d\eps,
\end{align*}
since for $l_{j+1}\leq \eps<l_j$,
the function $pV(\eps)$ is constant,
equal to $\sum_{k>j} l_k$.
We compute the integral to obtain
\begin{align*}
(1-s)\int_{l_n}^{l_1}pV(\eps)\eps^{s-2}\,d\eps
=\sum_{j=1}^{n-1}\sum_{k=j+1}^\infty l_k\bigl(l_{j+1}^{s-1}-l_j^{s-1}\bigr).
\end{align*}
Next,
we split the sum and interchange the order of summation,
to obtain
\[
\sum_{j=2}^{n}\sum_{k=j}^\infty l_kl_{j}^{s-1}-\sum_{j=1}^{n-1}\sum_{k=j+1}^\infty l_kl_j^{s-1}
=\sum_{k=2}^\infty l_k\sum_{j=2}^{\min\{k,n\}} l_{j}^{s-1}
-\sum_{k=2}^\infty l_k\sum_{j=1}^{\min\{k,n\}-1}l_j^{s-1}.
\]
In this formula,
the two double sums clearly converge,
since $\sum_{k\geq1}l_k$ converges.
Simplifying,
we obtain
\begin{align*}
(1-s)\int_{l_n}^{l_1}pV(\eps)\eps^{s-2}\,d\eps
&=\sum_{k=2}^\infty l_kl_{\min\{k,n\}}^{s-1}-\sum_{k=2}^\infty l_kl_1^{s-1}\\
&=\sum_{k=2}^{n-1} l_k^s+\sum_{k=n}^\infty l_kl_n^{s-1}-l_1^{s-1}\sum_{k=2}^\infty l_k.
\end{align*}

Now,
$l_kl_n^{s-1}\leq l_k^{s}$ for $k\geq n$ (if $s\geq1$,
we estimate instead $l_n^{s-1}\leq1$,
pro\-vi\-ded~$n$ is so large that $l_n\leq1$).
Hence we can let~$n$ approach infinity if and only if $\sum_{k=2}^\infty l_k^s$ converges,
and then the middle sum converges to zero.
In that case, we obtain $\sum_{k\geq1}l_k^s-l_1^{s-1}\zeta_\String(1)=\zeta_\String(s)-\zeta_\String(1)l_1^{s-1}$ for the limit.

In a similar way,
we compute
\begin{align*}
s\int_0^{l_n^{-1}} N(x)x^{-s-1}\,dx
=\sum_{j=1}^{n-1}s\int_{l_j^{-1}}^{l_{j+1}^{-1}}N(x)x^{-s-1}\,dx
=\sum_{j=0}^{n-1}j\bigl(l_j^s-l_{j+1}^s\bigr),
\end{align*}
since $N(x)=0$ for $x<l_1^{-1}$,
and $N(x)=j$ for $l_j^{-1}\leq x<l_{j+1}^{-1}$.
We find
\begin{align*}
s\int_0^{l_n^{-1}} N(x)x^{-s-1}\,dx
=\sum_{j=1}^{n-1}jl_j^s-\sum_{j=1}^{n}(j-1)l_{j}^s
=\sum_{j=1}^{n}l_j^s-nl_{n}^s.
\end{align*}
Now,
$nl_n^s\leq2\sum_{j=[n/2]}^nl_j^s$,
provided $s\geq0$,
so we can let $n$ approach infinity if and only if $\sum_{j=1}^\infty l_j^s$ converges,
in which case we find the value $\zeta_\String(s)$ for the limit,
again since the tail $\smash{\sum_{j=[n/2]}^\infty l_j^s}$ converges to zero.
\end{proof}

Recall that the Minkowski dimension $D_M$ was defined in~\eqref{dimension} above.
We also define the {\em growth rate} of $\mathcal{L}$ (or {\em asymptotic growth rate} of the geometric counting function $N := N_\mcl$) by 

\begin{gather}\label{E:dimension}
D:=\inf\bigl\{\alpha \geq 0~|~ N(x)=O(x^\alpha)\text{ as }x\to\infty\bigr\}.
\end{gather}

\begin{theorem}\label{T:p-adic D=sigma=D_M}
Assume that the hypotheses of Theorem \ref{T:p-adic D_M=sigma} are satisfied. Then $\sigma_{\mathcal{L}}$, the abscissa of convergence of\/ $\zeta_\String$, coincides with\/ $D_M$ and with\/~$D$. That is, $D_M = \sigma_\mathcal{L} = D$.
\end{theorem}

\begin{proof}
Let $\alpha > D_M.$ Since, by definition of $D_M, V(\eps)\leq A\eps^{1-\alpha}$,
then
\[ (1-s)\int_0^{l_1}V(\eps)\eps^{s-2}\,d\eps\leq A(1-s)\int_0^{l_1}\eps^{s-\alpha-1}\,d\eps. \] 
(Here, $A$ is some suitable positive constant.)
This integral converges for all real numbers $s>\alpha$;
hence,
by the foregoing lemma (Lemma \ref{L:zeta=N,V}),
$\sigma\leq\alpha$, where (for notational simplicity) $\sigma = \sigma_\mathcal{L}$ denotes the abscissa of convergence of $\mathcal{L} := \{\ell_j\}_{j=1}^\infty$.
Since this holds for all $\alpha \in \mathbb{R}$ such that $\alpha>D_M$,
we conclude that $\sigma\leq D_M$. Conversely,
if $\alpha<D_M$,
then $V(\eps)$ is not $O(\eps^{1-\alpha})$ as $\eps \rightarrow 0^+$.
This means that there exists a sequence $\{\eps_n\}_{n=0}^\infty$ converging to $0,$ with $l_1\geq\eps_0>\eps_1>\eps_2>\dots$ and such that $V(\eps_j)\geq\eps_j^{1-\alpha}$ for every $j \geq 1$.
Moreover,
we may choose the sequence to be exponentially decreasing; say,
$\eps_{n+1}<\eps_n/2$ for every $n \geq 1$.
Then,
for $s\leq1$,
\[
(1-s)\int_0^{l_1}V(\eps)\eps^{s-2}\,d\eps
\geq\sum_{j=1}^\infty(1-s)\int_{\eps_{j}}^{\eps_{j-1}}\eps_j^{1-\alpha}\eps^{s-2}\,d\eps,
\]
since $V(\eps)$ is increasing.
We then estimate
\[
(1-s)\int_{\eps_{j}}^{\eps_{j-1}}\eps^{s-2}\,d\eps
=\eps_{j}^{s-1}-\eps_{j-1}^{s-1}
\geq\eps_j^{s-1}\bigl(1-2^{s-1}\bigr),
\]
to obtain
\[
(1-s)\int_0^{l_1}V(\eps)\eps^{s-2}\,d\eps
\geq\sum_{j=1}^\infty\eps_j^{s-\alpha}\bigl(1-2^{s-1}\bigr).
\]
For all $s \in \mathbb{R}$ such that $s\leq\alpha$,
this sum diverges.
Again by Lemma \ref{L:zeta=N,V},
we conclude that $\sigma\geq\alpha$.
This holds for all $\alpha<D_M$; hence, $\sigma\geq D_M$.
Together with the first part,
we conclude that~$\sigma=D_M$.

Next, we show that $\sigma=D$.
If $\alpha>D$, then $N(x)\leq Ax^{\alpha}$ for some $A>0$.
Then for all $s \in \mathbb{R}$ such that $s>\alpha$, and
noting that $N(x)$ vanishes for $x<l_1^{-1}$, we have that
\[
s\int_{l_1^{-1}}^\infty N(x)x^{-s-1}\,dx\leq \frac{A}{s-\alpha}l_1^{s-\alpha}.
\]
Hence, according to Lemma \ref{L:zeta=N,V},
$\sigma\leq\alpha$.
Since this holds for all $\alpha>D$,
we conclude that $\sigma\leq D$. 

Conversely,
if $\alpha<D$,
then $N(x)$ is not $O(x^{\alpha})$.
This means that there exists an unbounded sequence $\{x_j \}_{j=0}^\infty$ tending to $\infty$, with $l_1^{-1}\leq x_0<x_1<x_2<\dots$ and such that $N(x_j)\geq x_j^{\alpha}$ for every~$j$.
Moreover,
we choose the sequence to be exponentially increasing,
$x_{n+1}>2x_n$.
Then,
for $s\geq0$,
\[
s\int_{l_1^{-1}}^\infty N(x)x^{-s-1}\,dx
\geq\sum_{j=0}^\infty s\int_{x_{j}}^{x_{j+1}}x_j^{\alpha}x^{-s-1}\,dx,
\]
since $N(x)$ is increasing.
We estimate
$
s\int_{x_{j}}^{x_{j+1}}x^{-s-1}\,dx
\geq x_{j}^{-s}\bigl(1-2^{-s}\bigr),
$
to obtain
\[
s\int_{l_1^{-1}}^\infty N(x)x^{-s-1}\,dx
\geq\sum_{j=1}^\infty x_{j}^{\alpha-s}\bigl(1-2^{-s}\bigr).
\]
For $s\leq\alpha$,
this sum diverges.
We conclude that $\sigma\geq\alpha$.
This holds for all $\alpha<D$,
hence $\sigma\geq D$,
and it follows that $\sigma=D$. 

Combining all of the above steps, we conclude that $\sigma = D_M = D,$ as desired.
\end{proof}

\begin{corollary}\label{C:0 D sigma 1}
For any\/ $p$-adic fractal string $\mathcal{L}_p$  with infinitely many lengths,
we have $0\leq D_M=\sigma_{\String_p}\leq1$. Furthermore, we have that $D_M = \sigma_{\mathcal{L}_p} = D,$ where $D$ is the growth rate of $\mathcal{L}_p$ defined by \eqref{E:dimension}.
\end{corollary}

\subsection{The Real Case}\label{TRC}

For (ordinary)
archimedean fractal strings,
the Minkowski dimension also determines the abscissa of convergence of the geometric zeta function,
in an analogous manner. The advantage of our new proof is that it yields a unified approach to both the archimedean and nonarchimedean (or $p$-adic) cases. It also establishes in the process the new result according to which $D_M$ is not only equal to $\sigma_\mcl$ (the abscissa of convergence of $\mcl$) but also to $D$ (the asymptotic growth rate of $\mcl$), a useful fact which was only implicit in earlier work (such as \cite{Lap, Lap3, LapPo, L-vF1, L-vF2}), even for real (or archimedean) fractal strings and is explicitly needed, for example, in \cite{HerLap, HerLap2, HerLap3}.

We provide here the details of our unified approach, but by focusing, of course, on the real (or archimedean) case.

The geometric zeta function of a real fractal string $\mcl = \{\ell_j \}_{j=1}^\infty,$ with $l_1\geq l_2\geq l_3\geq\dots\to0$, is given by~\eqref{E:zetaL},
just as in the $p$-adic case,
with an abscissa of convergence defined by
\begin{equation}\label{E:sigmaL}
\sigma_\String=\inf \Bigg\{ \alpha \in \mathbb{R} \colon \sum_{j=1}^\infty l_j^\alpha<\infty \Bigg\},
\end{equation}
entirely analogous to~\eqref{sigma}.
We assume an infinite number of positive lengths with a finite total length,
\[
\zeta_\String(1)=\sum_{j=1}^\infty l_j<\infty;
\]
so that $0\leq\sigma_\String\leq 1$.

Since both in the archimedean and in the nonarchimedean case,
the geometric  zeta function and geometric counting function $N_\String$ (and hence $D$) are defined in the same way,
it immediately follows that \eqref{E:zeta=N} holds in the archimedean case as well,
and consequently,
$\sigma_\String=D$.

On the other hand,
the formula for the volume of the tubular neighborhoods is different,
due to the different geometry of the boundary of balls in $\Q_p$ (as we have seen in \S\ref{inner tube}) and of intervals in $\R$.
In particular,
the real unit interval $[0,1]$ has inradius 1/2 in $\R=\Q_{\infty}$ whereas the $p$-adic  unit ball
$\Z_p$ has inradius 1 in $\Q_p$
(see also Remark~\ref{real vs. p-adic volume}).
For real fractal strings,
$V_\String(\eps)$ is given by~\eqref{LPformula},
and hence the formula corresponding to~\eqref{E:zeta=V} is
\begin{gather}\label{E:zeta=V arch}
\zeta_\String(s)=s\zeta_\String(1)l_1^{s-1}+2s(1-s)\int_0^{l_1/2}V_\String(\eps)(2\eps)^{s-2}\,d\eps,
\end{gather}
valid for $\Re s>D_M$,
where $D_M$ is defined by~\eqref{dimension}.
This can be proved by a method similar to the proof of Lemma~\ref{L:zeta=N,V},
but we give here an alternative proof.
The function $V_\String = V_\String(\eps)$ is continuous and piecewise differentiable for $\varepsilon > 0$,
with derivative
\[
V_\String'(\eps)=2N_\String\Bigl(\frac1{2\eps}\Bigr).
\]
Integrating by parts,
we obtain
\begin{align*}
2s(1-s)\int_0^{l_1/2}V_\String(\eps)&(2\eps)^{s-2}\,d\eps\\
&=-s\left[V_\String(\eps)(2\eps)^{s-1}\right]_0^{l_1/2}+2s\int_0^{l_1/2}N_\String(1/2\eps)(2\eps)^{s-1}d\eps\\
&=-s\zeta_\String(1)l_1^{s-1}+s\int_{l_1^{-1}}^\infty N_\String(x)x^{-1-s}\,dx,
\end{align*}
which by~\eqref{E:zeta=N} equals $-s\zeta_\String(1)l_1^{s-1}+\zeta_\String(s)$ since $N_\String(x)=0$ for $x<l_1^{-1}$.

\begin{remark}
As an alternative,
the zeta function is also given by
\[
\zeta_\String(s)=2s(1-s)\int_0^\infty V_\String(\eps)(2\eps)^{s-2}\,d\eps.
\]
This expression only converges for $D_M<\Re s<1$.
\end{remark}

As was already pointed out earlier,
the above formula converges for $\Re s>D_M$.
In addition,
it follows from the proof that the formula converges if and only if the series for $\zeta_\String$ converges.
This implies that $\sigma_\String\leq D_M$.
In order to prove the converse inequality,
we construct,
just as in the proof of Theorem~\ref{T:p-adic D=sigma=D_M},
 a sequence of positive $\eps$-values decreasing exponentially fast to zero in order to show that~\eqref{E:zeta=V arch} does not converge for $\Re s<D_M$.
It follows that $\sigma_\String=D_M$.

We conclude that for a real (or archimedean) fractal string, we have $D_M = \sigma_\mathcal{L} = D$, 
just as was shown in the first part of this section for a $p$-adic (or nonarchimedean) fractal string, 
and thereby completing the statement and the proof of Theorem \ref{T:p-adic D_M=sigma} and Theorem \ref{T:p-adic D=sigma=D_M} (now extended to the real case), 
as well as providing a unified treatment of both the archimedean and nonarchimedean cases.

\section{Explicit Tube Formulas for $p$-adic Fractal Strings} \label{explicit tf} 

The following result is the counterpart in this context of Theorem 8.1 of~\cite{L-vF2},
the distributional tube formula for real fractal strings.  It is established by using, in particular, the extended distributional explicit formula of~\cite[Thms.~5.26 and~5.27]{L-vF2},
along with the expression for the volume of thin inner $\eps$-tubes obtained in Theorem \ref{thin}.

We now state our general nonarchimedean (or $p$-adic) fractal tube formula in this context.

\begin{theorem}[$p$-adic explicit tube formula]\label{dtf}
$($i$)$
Let $\String_p$ be a languid $p$-adic fractal string $($as in the first part of Definition \ref{languid}$)$, for some real exponent $\kappa$ and a screen $S$ that lies strictly to the left of the vertical line $\Re(s)=1$.  
Further assume that $\sigma_{\String_p}<1.$ $($Recall from Corollary~\ref{C:0 D sigma 1} that we always have $\sigma_{\String_p}\leq 1$.
Moreover, if $\String_p$ is self-similar, then $\sigma_{\String_p}<1.)$
Then the volume of the thin inner  $\eps$-neighborhood of $\String_p$ is given by the following distributional explicit formula, on test functions in $\mathbf{D}(0, \infty)$, the space of $C^{\infty}$ functions with compact support in $(0,\infty)\colon$
\begin{equation}\label{detf}
V_{\String_p}(\eps)= \sum_{ \omega \in \mathcal{D}_{\mathcal{L}_p} (W)}\res 
\left(
\frac{p^{-1} \zeta_{\mathcal{L}_p} (s) \eps^{1-s}} {1-s}; \omega
\right) 
+ \mathcal{R}_p (\eps),
\end{equation}
where $\mathcal D_{\String_p}(W)$ is the set of visible complex dimensions of $\String_p$ $($as given in Definition \ref{dvcd}$).$
Here, the distributional error term is given by 
\begin{equation}\label{error term}
\mathcal R_p(\eps)=\frac{1}{2\pi i}
\int_{S}\frac{p^{-1} \zeta_{\mathcal{L}_p} (s) \eps^{1-s}} {1-s}ds
\end{equation}
and is estimated distributionally $($in the sense of {\em \cite[Defn.~5.29]{L-vF2})} by
\begin{equation}\label{estimate}
\mathcal{R}_p(\eps)=O(\eps^{1-\sup S}), \quad \mbox{\em as} ~\eps \rightarrow 0^+.
\end{equation}

$($ii$)$ Moreover, if $\String_p$ is strongly languid $($as in the second part of Definition \ref{languid}$)$, then we can take $W=\Com$ and
$\mathcal{R}_p(\eps)\equiv 0$, provided we apply this formula to test functions supported on   compact subsets  of $[0, A)$. The resulting explicit formula without error term is often called an \emph{exact tube formula} in this case. 
\end{theorem}

\begin{proof}
Since the proof of Theorem~\ref{dtf} parallels that of its counterpart for real fractal strings (see~\cite[Thm.~8.7]{L-vF2}),
 we only provide here the main steps.
We will explain,
 in particular,
 why the $p$-adic tube formula takes a different form than in the real case.
As will be clear from the proof,
 it all goes back to the difference between Theorem~\ref{thin} and its archimedean analogue (see Equation \eqref{LPformula} above or \cite[Equation~(8.1)]{L-vF2}).

According to Theorem~\ref{thin}, 
 \begin{align*} 
V_{\String_p}(\eps)=\frac1p\sum_{j\colon  p^{-n_j}< \eps} p^{-n_j}
=\frac1p\int_{1/\eps}^{\infty} \frac{1}{x}\,\eta(dx)
=\langle\mathcal P_{\eta}^{[0]}, v_{\eps}\rangle,
  \end{align*}
  where $\mathcal P_{\eta}^{[0]} =   \eta:=\sum_{j=1}^{\infty}\delta_{ \{ {p^{n_j}} \}}$ is viewed as a distribution   and 
  \[
v_{\eps}(x):=\begin{cases}
0		&\mbox{ if } x\leq 1/{\eps} \\
1/({px})	&\mbox{ if } x>1/{\eps}.
\end{cases}
\]

Fix $\varphi \in  \mathbf {D}(0,\infty)$.
Then
 \begin{align*}
\int_0^{\infty}\varphi(\eps)v_{\eps}(x)\,d\eps
&=\frac1{px}\int_{1/x}^{\infty}\varphi(\eps)\,d\eps \\
&=p^{-1}\varphi_1(x),
\end{align*} 
where $\varphi_1$ is a smooth,
but not compactly supported,
 test function,
 given by 
\[
\varphi_1(x):=\frac{1}{x}\int_{1/x}^{\infty}\varphi(\eps)d\eps.
\]
Thus
 \begin{align}
\langle V_{\String_p}(\eps), \varphi \rangle
&=\int_0^{\infty}\varphi(\eps) \int_0^{\infty}v_{\eps}(x)\,\eta(dx) d\eps\nonumber\\
&=\Bigl\langle \mathcal P_{\eta}^{[0]},p^{-1}\varphi_1(x)\Bigr\rangle.\label{e3}
\end{align}

The Mellin transform of $\varphi_1$ is computed to be
 \begin{gather}
\tilde{\varphi}_1(s)=\frac{1}{1-s}\,\tilde{\varphi}(2-s)\quad\text{ for }\Re s<1\label{e2}. 
\end{gather}
Furthermore,
 by analytic continuation,
 and since $ \tilde{\varphi}(s)$ is entire for
 $\varphi\in   \mathbf {D}(0,\infty) $,
 the equality in \eqref{e2} continues to hold for all  $s\in \Com.$

Now,
 let $ \Psi =p^{-1}\varphi_1$.
Its Mellin transform  is 
\[
\widetilde{\Psi}(s)=\frac{p^{-1}}{1-s}\,\tilde{\varphi}(2-s),
\]
which holds for all $s\in \Com$.
Note that it follows from our previous discussion that $\widetilde{\Psi}(s)$ is meromorphic in all of $\Com$,
 with a single,
 simple pole at $s=1$. 

Next, we deduce from \eqref{e3} and~\cite[Thm.~5.26]{L-vF2} (the extended distributional explicit formula) that
\begin{align*}
\langle V_{\String_p}(\eps), \varphi \rangle
&=  \sum_{ \omega \in \mathcal D_{\String_p}}
 \res\left(  \zeta_{\String_p}(s) \widetilde{\Psi}(s); \omega\right) + 
 \Rerror_p (\eps) \\
&=\int_0^{\infty} \sum_{ \omega \in \Dimensions_{\String_p}} \res \left(\frac{\zeta_{\String_p} (s)\eps^{1-s}} {p(1-s)}; \omega\right)\varphi(\eps)d\eps +\int_0^{\infty}\Rerror_p (\eps)\varphi(\eps)\,d\eps.
\end{align*}
Therefore, 
\[
V_{\String_p}(\eps)= \sum_{ \omega \in \mathcal D_{\String_p} (W)} \res \left (\frac{\zeta_{\String_p} (s) \eps^{1-s}} {p(1-s)}; \omega\right)+ \Rerror_p (\eps),
\]
where the distribution $\Rerror_p (\varepsilon)$ is given by \eqref{error term} and is estimated distributionally as in \eqref{estimate}. 

In closing this proof,
 we note that in the strongly languid case,
 we use~\cite[Thm.~5.27]{L-vF2} in order to conclude that (\ref{dtf}) holds with $\Rerror_p(\eps)\equiv0.$ 
 \end{proof}

\begin{remark}\label{tzf}
We may rewrite the (typically infinite) sum in (\ref{detf}) as follows:
\begin{equation}\label{sum residue}
\sum_{ \omega \in \mathcal{D}_{\mathcal{L}_p} (W)} \res(
\zeta_{\String_p}(\eps; s); s=\omega), 
\end{equation}
where (by analogy with the definitions and results in~\cite{LapPe1,LPW}),
\begin{equation}\label{tubular zeta}
\zeta_{\String_p}(\eps; s):=\frac{p^{-1}\zeta_{\String_p}(s)\eps^{1-s}}{1-s}
\end{equation}
is called the \emph{nonarchimedean tubular zeta function} of the $p$-adic fractal string $\String_p$.

By contrast,
 the archimedean tubular zeta function (in the present one-di\-men\-sional situation) of a real fractal string $\String$ is given by 
\begin{equation}\label{real tubular zeta}
\zeta_{\String}(\eps; s):=\frac{\zeta_{\String}(s)(2\eps)^{1-s}}{s(1-s)},
\end{equation}
and the analog of the above sum in the archimedean tube formula of~\cite{L-vF2} (as rewritten in~\cite{LapPe1}) is given as in~\eqref{sum residue},
except with $\String_p$ replaced by $\String$ and with 
$\mathcal D_{\String}(W)\cup \{0\}$ instead of $\mathcal D_{\String_p}(W)$.
Note that $\zeta_{\String}(\eps; s)$ typically has a pole at $s=0$,
whereas
$\zeta_{\String_p}(\eps; s)$ does not.
\end{remark}

\begin{corollary}[$p$-adic fractal tube formula]\label{ftf}
 If, in addition to the hypotheses in Theorem \ref{dtf}, we assume that all the visible complex dimensions of 
 $\mathcal{L}_p$ are  simple, then  
\begin{equation}\label{pftf} 
V_{\String_p}(\eps)= \sum_{ \omega \in \mathcal{D}_{\String_p}(W)} 
c_{\omega} \frac{ \eps^{1-\omega}} {1-\omega } 
+ \mathcal{R}_p (\eps),
\end{equation}
where $c_{\omega}=p^{-1}\res\left(\zeta_{\String_p}; \omega\right)$. 
Here, the error term $\mathcal R_p$ is given by (\ref{error term})
and is estimated by (\ref{estimate}) in the languid case. 
Furthermore, we have $ \mathcal{R}_p (\eps) \equiv 0$ in the strongly languid case, provided we choose $W=\Com$. 
\end{corollary}
 
 \begin{remark}
 In~\cite[Ch.~8]{L-vF2}, under different sets of assumptions, both distributional and pointwise tube formulas are obtained for archimedean fractal strings (and also, for archimedean self-similar fractal strings). 
 (See, in particular, Theorems 8.1 and 8.7, along with \S8.4 in~\cite{L-vF2}.)
 At least for now, in the nonarchimedean case, we limit ourselves to discussing distributional explicit tube formulas. We expect, however, that under appropriate hypotheses, one should be able to obtain a pointwise fractal tube formula for $p$-adic fractal strings and especially, for $p$-adic self-similar strings. 
 In fact, for the simple examples of the nonarchimedean Cantor, Euler and Fibonacci strings, the direct derivation of the fractal tube formula (\ref{pftf})
 yields a formula that is valid pointwise and not just distributionally.
(See, in particular,
\S\ref{euler volume} and Example~\ref{cantor volume}.)
We leave the consideration of such possible pointwise extensions to a future work.
 \end{remark}
 
 \begin{example}[Fractal tube formula for the $p$-adic Euler string]\label{eulerrevisit}
 We now explain how to recover from Theorem \ref{dtf} (or Corollary \ref{ftf}) the tube formula for the Euler string $\mathcal E_p$ obtained via a direct computation in \S\ref{euler volume}.
 Indeed, it follows from Corollary \ref{ftf} (applied with $W=\Com$)
  that 
 \begin{equation}\label{etf}
 V_{\mathcal E_p}(\eps)= \frac{1}{p}\sum_{\omega \in \mathcal D_{\mathcal E_p} }
 \res(\zeta_{\mathcal E_p}; \omega)
   \frac{\eps^{1-\omega}}{ 1-\omega },
 \end{equation}
 which is exactly the expression obtained for $ V_{\mathcal E_p}(\eps)$ in formula 
 (\ref{VolumeEuler}) of \S\ref{euler volume} since
\[
\res(\zeta_{\mathcal E_p}; \omega)=\frac{1}{\log p}
\]
 for all $\omega \in \mathcal D_{\mathcal E_p}.$
  (This latter observation follows easily from the expression of $\zeta_{\String_p}$ obtained in Equation~\eqref{ptheuler}.)
  Note that Corollary \ref{ftf} can be applied here in the strongly languid case when $W=\Com$ and 
  $ \mathcal{R}_p (\eps) \equiv 0$  since, in light of the discussion in \S\ref{Euler string}, all the complex dimensions of $\mathcal E_p$ are simple and $\zeta_{\mathcal E_p}$ is clearly strongly languid of order $\kappa: =0$ and with the constant $A:=p^{-1}.$
  Furthermore, formula (\ref{etf}) can be rewritten in the following more concrete form:
   \begin{equation}\label{cetf}
 V_{\mathcal E_p}(\eps)= \frac{1}{p\log p}\sum_{n\in \Z}
   \frac{\eps^{1-in\mathbf{p}}}{ 1-in\mathbf{p} },
 \end{equation}
 since
$ \mathcal D_{\mathcal E_p}=\{in\mathbf{p} : n\in \Z\}$ and $\mathbf{p}=2\pi/\log p$ 
  (as in Equation (\ref{cdes}) of \S\ref{Euler string}).
  
  Finally, note that since the  series 
\[
  \sum_{n\in \Z}
   \frac{\eps^{1-in\mathbf{p}}}{ 1-in\mathbf{p} }
\]
   converges pointwise because the associated Fourier series 
   $ \sum_{n\in \Z}
   \frac{e^{2\pi in x}}{ 1-in\mathbf{p} }$ 
   is pointwise convergent on $\R$, it follows that the $p$-adic fractal tube formulas (\ref{etf})--(\ref{cetf}) actually converge pointwise rather than just distributionally. 
  \end{example}

\begin{example}[The tube formula for the nonarchimedean Cantor string]\label{cantor volume}
In this example,
 we explain how to derive the exact fractal tube formula for $\mathcal{CS}_3$, the 3-adic Cantor string introduced in \cite{LapLu1} and further studied in \cite{LapLu2,LapLu2.5,LapLu3}.

By construction, the complement of $\mathcal{CS}_3$ in $\Z_3$ is the 3-adic Cantor set $\mcc_3$, which is a nonarchimedean self-similar set (as introduced in \cite{LapLu1, LapLu2}); so that $\mcc_3$ is the unique nonempty compact subset $K$ of $\Z_3$ (or of $\Q_3$) which is the solution of the fixed point equation

\begin{equation}\label{6.12}
K= \varphi_1(K) \cup \varphi_2 (K),
\end{equation}
for some suitable affine similarity transformations $\varphi_1, \varphi_2$ from $\Z_3$ to itself; more specifically, we have that $\varphi_1 (x) = 3x$ and $\varphi_2 (x) = 2 + 3x$, for all $x \in \Z_3$. We refer to \cite{LapLu1,LapLu2,LapLu2.5} for more information concerning the properties of $\mcc_3$ and $\mathcal{CS}_3$ as well as for corresponding figures. (See \cite{Hut} for the general definition of self-similar sets in complete metric spaces, and \cite{Fal} for a detailed discussion in the usual case of Euclidean spaces.)

Let $\eps >0$. We have that

\[ \zeta_{\mathcal{CS}_3} (s) = \frac{3^{-s}}{1 - 2\cdot3^{-s}}, \text{ for all } s \in \Com\] and hence 

\[ \Dimensions_{\mathcal{CS}_3} = \{ D + i \nu \textbf{p}~ |~ \nu \in \Z\}, \]
will all the complex dimensions being simple and where $D:= \log_3 2$ and $\textbf{p} := 2 \pi/\log3 $. Furthermore, we have that

\[
\res(\zeta_{\mathcal{CS}_3}; \omega)=\frac{1}{2\log3},
\]
 independently of $\omega \in \mathcal{D}_{\mathcal{CS}_3 }$,
 and so the exact fractal tube formula for the nonarchimedean Cantor string is found to be 
\begin{equation}\label{tfc}
V_{\mathcal{CS}_3}(\eps)=
\frac{1}{ 3} \sum_{\omega \in \mathcal D_{\mathcal {CS}_3}}\res(\zeta_{\mathcal{CS}_3}; \omega)\frac{\eps^{1-\omega  }}{1-\omega}.
\end{equation}

Note that since $\mathcal{CS}_3$ has simple complex dimensions,
 we may also apply Corollary \ref{ftf} (in the strongly languid case when $W=\Com$) in order  to precisely recover Equation~\eqref{tfc}.

We may rewrite~\eqref{tfc} in the following form:
\[
V_{\mathcal{CS}_3}(\eps)=\eps^{1-D}G_{\mathcal{CS}_3}(\log_{3}\eps^{-1}),
\]
where $G_{\mathcal{CS}_3}$ is the nonconstant periodic function (of period 1) on $\R$ given by 
\[
G_{\mathcal{CS}_3}(x):=\frac{1}{6\log3}\sum_{n\in \Z}\frac{e^{2\pi inx}}{1-D-in\mathbf{p}}.
\]
Finally,
 we note that since the Fourier series 
\[
\sum_{n\in \Z}\frac{e^{2\pi inx}}{1-D-in\mathbf{p}}
\]
is pointwise convergent on $\R$,
 the above direct computation of $V_{\mathcal{CS}_3}(\eps)$
shows that~\eqref{tfc} actually holds pointwise rather than just distributionally.

In closing this example, we note that we could similarly use Theorem \ref{dtf} (or Corollary \ref{ftf}) to obtain an exact fractal tube formula for the $5$-adic Cantor string recently introduced in \cite{KRC} and defined in a way analogous to the $3$-adic Cantor string from \cite{LapLu1}.
\end{example}

\begin{remark}\label{R6.6.1/2}
The $3$-adic Cantor string discussed in Example \ref{cantor volume} is an example of a $p$-adic (here, 3-adic) self-similar string. Another example of a $p$-adic (or nonarchimedean) self-similar string is the 2-adic Fibonacci string, whose complex dimensions are distributed periodically along two vertical lines (instead of a single one as in the case of a 3-adic Cantor string). (See \cite{LapLu2,LapLu2.5}; furthermore, see \cite{LapLu3} for the corresponding exact pointwise tube formula.) In general, a (nontrivial) $p$-adic self-similar string $\mcl_p$ is always lattice (that is, its scaling ratios are all integer powers of a single number, necessarily $p$; see \cite{LapLu2, LapLu2.5}. Therefore, unlike for real (or archimedean) fractal strings (compare with \cite[Chs. 2 \& 3]{L-vF2}), which can be either lattice or nonlattice, the complex dimensions of $\mcl_p$ are always periodically distributed along finitely many vertical lines, the right most of which is the vertical line $\{ \mfr s = D_M \}$, where $D_M$ is the Minkowski dimension of $\mcl_p$. The corresponding fractal tube formulas, illustrating our main theorem in this section (Theorem \ref{dtf}) in order to obtain fractal tube formulas for general $p$-adic self-similar fractal strings, are provided in \cite{LapLu3}.

In order to avoid unnecessary repetitions, we refer the interested reader to \cite{LapLu3} (and \cite{LapLu2.5}) for those special but important examples of fractal tube formulas. We only mention the following two interesting facts: 
\begin{itemize}
\item[(i)] Because on each relevant vertical line, the complex dimensions form an arithmetic progression (with a progression or period independent of the line) and have the same multiplicities, the corresponding term in the associated fractal tube formula can be written as a suitable power function times a periodic function (of $x := \log (\varepsilon^{-1})$). (This is so assuming that the complex dimensions on that line are simple, which is always the case, for instance, of the right most vertical line $\{\mfr s = D_M \}$.)
\item[]
\item[(ii)] In all of the concrete examples of $p$-adic self-similar strings studied in \cite{LapLu2.5, LapLu3}, including the 3-adic Cantor string and the 2-adic Fibonacci string, the corresponding exact fractal tube formula can be shown to converge {\em pointwise} (rather than distributionally, as in 
Theorem \ref{dtf}). We conjecture that at least in the case of simple complex dimensions, the exact fractal tube formula of a $p$-adic self-similar string always converges pointwise (and not just distributionally, as in Theorem \ref{dtf}). (Such a result is established in \cite[\S 8.4]{L-vF2} for general real or archimedean self-similar strings, whether or not all of the complex dimensions are simple.) Accordingly, it would be very interesting to establish that conjecture as well as to obtain a pointwise counterpart of Theorem \ref{dtf}; that is, a fractal tube formula for $p$-adic (not necessarily self-similar) fractal strings, with or without an error term, which (under suitable hypotheses) would be valid pointwise. We note that in the archimedean case (i.e., for real fractal strings) such a pointwise fractal tube formula is available under rather general conditions; see \cite[\S 8.1.1, esp., Th. 8.7 \& Cor. 8.10]{L-vF2}. We leave the investigation of these issues to some future work or to the interested reader.
\end{itemize}
\end{remark}

\section{Possible Extensions}
\label{towards}

We close this paper by providing possible directions for future investigations in this area. In Remark \ref{R6.6.1/2}, we have already mentioned the problem of obtaining a pointwise fractal tube formula, analogous to our distributional fractal tube formula (Theorem \ref{dtf}) in the nonarchimedean case and to the pointwise tube formula obtained in the archimedean case in \cite[\S 8.1.1]{L-vF2} (for general real fractal strings, under suitable hypotheses) and (without any assumptions) in \cite[\S 8.4]{L-vF2} for general real self-similar fractal strings. We next point out other possible problems and research directions.

 \subsection{Ad\`elic Fractal Strings and Their Spectra} \label{adelic}

It would be interesting to unify the archimedean and nonarchimedean settings by appropriately defining \emph{ad\`elic} fractal strings, and then studying the associated spectral zeta functions (as is done for standard archimedean fractal strings in~\cite{Lap,Lap3} and~\cite{LapMa,LapPo,LapPo3,L-vF1,L-vF2}).
To this aim,
the spectrum of these ad\`elic fractal strings should be suitably defined and its study may benefit from Dragovich's work~\cite{Drag} on ad\`elic quantum harmonic oscillators. In the process of defining these ad\`elic fractal strings, we expect to make contact with the notion of a fractal membrane (or ``quantized fractal string'') introduced in~\cite[Ch.~3]{Lap2} and rigorously constructed in~\cite{LapNe} as a Connes-type noncommutative geometric space \cite{Con}; see also~\cite[\S 4.2]{Lap2}.
The aforementioned spectral zeta function of an ad\`elic fractal string would then be viewed as the (completed) spectral partition function of the associated fractal membrane,
 in the sense of~\cite{Lap2}.
(See also Remark~\ref{harmonic string} above.)
We note that a geometric construction of certain ad\`elic fractal strings is proposed in the epilogue (\S8) below.

\subsection{Nonarchimedean Fractal Strings in Berkovich Space} \label{berkovich}

    As was shown in \cite{LapLu2} and recalled in Remark \ref{R6.6.1/2}, there can only exist \emph{lattice} $p$-adic  self-similar strings, because of the discreteness of the valuation group of $\Q_p$. However, in the archimedean setting, there are both lattice and nonlattice self-similar strings; see \cite[Chs. 2 \& 3]{L-vF2}. We expect that by suitably extending the notion of  
    $p$-adic self-similar  string to Berkovich's $p$-adic analytic space~\cite{Ber,Duc},
it can be shown that $p$-adic  self-similar strings are generically nonlattice in this broader setting.
Furthermore, we conjecture that every nonlattice string in the Berkovich projective line  can be approximated by lattice strings with increasingly large oscillatory periods (much as occurs in the archimedean case~\cite[Ch.~3]{L-vF2}).
Finally, we expect that, by contrast with what happens for $p$-adic fractal strings, the volume 
$V_{\String_p}(\eps)$ will be a continuous function of $\eps$ in this context. (Compare with Remark \ref{real vs. p-adic volume}.)

\subsection{Higher-Dimensional Fractal Tube Formulas} \label{higher}

We expect that the higher-dimensional tube formulas obtained by Lapidus and Pearse in~\cite{LapPe1,LapPe2} (as well as, more generally, by those same authors and Winter in~\cite{LPW}) for archimedean self-similar sprays and the associated tilings~\cite{Pe} in $\R^d$ have a natural nonarchimedean counterpart in the $d$-dimensional $p$-adic space $\Q_p^d,$ for any integer $d\geq 1$. In the latter $p$-adic case, the corresponding `tubular zeta function'
 $\zeta_{\mathcal T_p}(\eps; s)$ (when $d=1$, see Remark \ref{tzf}) should have a more complicated expression than in the one-dimensional situation, and should involve both the inner radii and the `curvature' of the generators (see~\cite{LapPe1,LPW}, as described in \cite[\S 13.1]{L-vF2}, for the archimedean case) of the tiling (or $p$-adic fractal spray) $\mathcal T_p.$
Moreover,
by analogy with what is expected to happen in the Euclidean case~\cite{LapPe1,LPW},
the coefficients of the resulting higher-dimensional tube formula should have an appropriate interpretation in terms of yet to be suitably defined  `nonarchimedean fractal curvatures' associated with each complex and integral dimension of $\mathcal T_p$.  Finally, by analogy with the archimedean case (for $d\geq1$,
see~\cite{LapPe1} and~\cite{LPW}),
the $p$-adic higher-dimensional fractal tube formula should take the same form as in Equation~\eqref{sum residue},
except with $\zeta_{\String_p}(\eps; s)$ given by a different expression from the one in (\ref{tubular zeta}) where $d=1$,
 and with $\mathcal D_{\String_p}(W)$ replaced by 
 $ \mathcal D_{\String_p}(W)\cup \{0, 1, \ldots, d\}$, as well as (for nonarchimedean self-similar tilings) with $W=\Com$ and $\mathcal R_p(\eps)\equiv 0$ in the counterpart of Equation (\ref{tubular zeta}) or (\ref{real tubular zeta}). 
 In the future, we plan to investigate the above problems along with related question pertaining to fractal geometry and geometric measure theory in nonarchimedean spaces.

In closing \S\ref{higher}, we mention that  recently, the first author, Goran Radunovi\` c and Darko \u Zubrinic have developed a general theory of fractal zeta functions and complex dimensions  valid in Euclidean spaces $\R^N$ of any dimension and for arbitrary bounded subsets of $\R^N$ (see, e.g., the  book \cite{LapRaZu1}). In the process, they have very significantly extended the theory of fractal tube formulas obtained originally for fractal strings in \cite{L-vF1, L-vF2} and then for higher-dimensional fractal sprays (especially, self-similar sprays) in \cite{LapPe1, LapPe2, LPW}; see, especially, \cite{LapRaZu2, LapRaZu3} and \cite[Ch. 5]{LapRaZu1}. Accordingly, it is natural to wonder whether the general theory of fractal zeta functions and fractal tube formulas developed in \cite{LapRaZu1} and \cite{LapRaZu2, LapRaZu3,LapRaZu3.5,LapRaZu4} can be applied and suitably adapted in order to obtain concrete nonarchimedean tube formulas valid (under appropriate hypotheses) for arbitrary compact subsets of $p$-adic space $(\Q_p)^N$ (or more general ultrametric spaces), and, in particular, for arbitrary $p$-adic self-similar sets in $(\Q_p)^N$.

\section{Epilogue}\label{epilogue}
In looking for a simple geometric way to create an ad\`elic fractal string and a global theory of complex fractal dimensions, 
we found a very natural construction of $p$-adic fractal strings of any rational dimension between 0 and 1. 
The simplest example is of dimension $D=\frac{1}{2}$, which is particularly interesting since it involves the diagonal of the digits. 
This reminds one of the intersections of the graph of the Frobenius with the diagonal in Enrico Bombieri's proof of the Riemann hypothesis for curves over finite fields; see \cite{Bom, MvF}.
It may give rise to a fractal approach to translating his proof for curves over finite fields to the curve spec ${\mathbb Z}$ over the rationals, which is the case of the famous Riemann hypothesis for the Riemann zeta function. 
However, we caution the reader that this possibility is far from being realized for now. 

We found another natural way to create an infinite family of $p$-adic Cantor strings $\mathcal{CS}_p$ in the nonarchimedean ring of $p$-adic integers $\mathbb Z_p$ and simultaneously their exact counterparts in the archimedean unit interval $[0,1]$, the $p$-inary Cantor strings $\mathcal{CS}_p^*$.
The Minkowski dimensions of the nonarchimedean and archimedean Cantor strings vary  from 0 to 1 as $p$ varies from 2 to $\infty.$
Directly above and below the Minkowski dimension  lie infinitely many complex fractal dimensions, periodically distributed along a discrete vertical line. 
The periodic distribution of the complex fractal dimensions, being discrete near dimension $0$, become denser as the Minkowski dimension tends to 1.

The simplest way to unify all infinitely many $p$-adic Cantor strings $\mathcal{CS}_p$ together with the ordinary real Cantor string $\mathcal{CS}$ is to form an  infinite product 
\[\mathcal{CS}\times \prod_{p<\infty} \mathcal{CS}_p,\]
which is a self-similar string in the set of integral ad\`eles $\mathbb A_{\mathbb Z}$.

An even more harmonious and symmetric way to unify all the nonarchimedean Cantor strings together with their corresponding archimedean counterpart is first to pair each $p$-adic Cantor string $\mathcal{CS}_p$ together with the $p$-inary Cantor string $\mathcal{CS}_p^*$ by taking the Cartesian product 
$\mathcal{CS}_p\times\mathcal{CS}_p^*\subset \mathbb{Q}_p \times \mathbb R$. Then we can imagine the  infinite direct product 
\[
\prod_{p< \infty} (\mathcal{CS}_p\times\mathcal{CS}_p^*)\]
as being an `ad\`elic' Cantor string in a new `ad\`elic' space 
\[
\prod_{p< \infty} (\mathbb Q_p\times \mathbb R)
\]
with infinitely many archimedean components.\footnote
{We put `ad\`elic' in quotes because this space has infinitely many real components,
one for every prime number,
whereas the ring of ad\`eles has only one real component,
corresponding to the archimedean valuation of $\mathbb Q$.}

We note, however, that our constructions of ad\`elic fractal strings do not give the Riemann zeta function as the geometric zeta function. 
It would be interesting to have a natural construction of an ad\`elic fractal string with the Riemann zeta function as its geometric zeta function.

We conclude these comments with a construction that gives the square of the Riemann zeta function.  
Let $\mathcal E_p$ be the $p$-adic Euler string and $h$ be the real harmonic string, 
then the infinite direct product 
\[h\times \prod_{p<\infty}\mathcal E_p\]
can be considered as an ad\`elic fractal string in the set of integral ad\`eles $\mathbb A_{\mathbb Z}$. 
Let $\zeta_{\mathcal E_p}$ be the geometric zeta function of  $\mathcal E_p$ and $\zeta_h$ be the geometric zeta function of the harmonic string; then, the infinite product of complex meromorphic functions
\[
\zeta_h\times
\prod_{p<\infty}\zeta_{\mathcal E_p}
\]
is equal to the square of the Riemann zeta function. 


\section*{Acknowledgments} 

We wish to thank Springer, the publisher of \cite{L-vF2},  for having granted us 
(more specifically, the two authors of the book \cite{L-vF2})
the copyright for \cite[Section 13.2]{L-vF2} within which part of the results obtained in this paper were discussed.
That portion of \cite[Section 13.2]{L-vF2}
was referring,
 in particular,
  to an earlier preprint of this article to which we have since then made a number of changes and additions,
 including the new result providing 
 (in \S\ref{Mdimension}), among other things, 
 a full synthetic proof of the equality of the abscissa of convergence of the geometric zeta function and of the (upper) Minkowski dimension of the associated fractal set, valid both in the real (or archimedean) case and in the 
$p$-adic (or nonarchimedean) case.
 Another significant addition to our earlier version of this paper is the discussion (in \S\ref{epilogue}) of several results concerning the new topic of ad\`elic fractal strings, for which complete proofs will be provided in a later work towards a global theory of complex fractal dimensions.

\bibliographystyle{amsplain}

\end{document}